\documentclass[11pt,letterpaper]{article}

\usepackage[utf8]{inputenc}
\usepackage[bookmarks,colorlinks,breaklinks]{hyperref}
\hypersetup{urlcolor=blue, colorlinks=true, citecolor=green!50!black, linkcolor=blue}
\usepackage{multirow}
\usepackage[T1]{fontenc} 
\usepackage{amsmath}
\usepackage{amsfonts}
\usepackage{mathtools}
\usepackage{amssymb}
\usepackage{amsthm}
\usepackage{thmtools}
\usepackage{thm-restate}
\usepackage{pgffor}
\usepackage{xcolor}
\usepackage[lined,boxed,ruled,norelsize,algo2e,linesnumbered,noend]{algorithm2e}
\usepackage{cleveref}
\usepackage{graphicx}
\usepackage{geometry}
\usepackage{enumitem}
\usepackage{url}
\usepackage{tikz}
\usepackage{booktabs}
\usepackage{caption}  
\usepackage{txfonts}  
\usepackage{lmodern} 
\usepackage{microtype}

\declaretheorem[refname={Theorem,Theorems},Refname={Theorem,Theorems},numberwithin=section]{theorem}
\declaretheorem[numberlike=theorem]{lemma}

\declaretheorem[numberlike=theorem]{corollary}

\declaretheorem[numberlike=theorem]{observation}
\declaretheorem[numbered=no]{remark}

\newcommand{\Z}{\mathbb{Z}}
\newcommand{\F}{\mathbb{F}}
\renewcommand{\tilde}{\widetilde}
\renewcommand{\hat}{\widehat}
\renewcommand{\bar}{\overline}

\DeclareMathOperator*{\MM}{\operatorname{MM}}

\newcommand{\poly}{\operatorname{poly}}

\newcommand{\rank}{\operatorname{rank}}

\foreach \x in {A,...,Z}{%
	\expandafter\xdef\csname m\x\endcsname{\noexpand\mathbf{\x}}
}

\foreach \x in {A,...,Z}{%
	\expandafter\xdef\csname om\x\endcsname{\noexpand\overline{\noexpand\mathbf{\x}}}
}

\foreach \x in {A,...,Z}{%
	\expandafter\xdef\csname c\x\endcsname{\noexpand\mathcal{\x}}
}

\newcommand{\otau}{\noexpand{\overline{\tau}}}
\newcommand{\omu}{\noexpand{\overline{\mu}}}
\newcommand{\osigma}{\noexpand{\overline{\sigma}}}%
\newcommand{\os}{\noexpand{\overline{s}}}
\newcommand{\ot}{\noexpand{\overline{t}}}
\newcommand{\og}{\noexpand{\overline{g}}}
\newcommand{\ou}{\noexpand{\overline{u}}}
\newcommand{\ov}{\noexpand{\overline{v}}}
\newcommand{\ow}{\noexpand{\overline{w}}}
\newcommand{\ox}{\noexpand{\overline{x}}}
\newcommand{\oy}{\noexpand{\overline{y}}}
\newcommand{\oz}{\noexpand{\overline{z}}}
\newcommand{\oc}{\noexpand{\overline{c}}}

\makeatletter
\renewcommand{\paragraph}{%
	\@startsection{paragraph}{4}%
	{\z@}{1.25ex \@plus 1ex \@minus .2ex}{-1em}%
	{\normalfont\normalsize\bfseries}%
}
\makeatother

\usepackage{relsize}
\newcommand*{\defeq}{\stackrel{\mathsmaller{\mathrm{def}}}{=}}

\usepackage{tabularx}  
\newenvironment{problem}[1]{
  \vskip 1em \small \begin{tabular}{@{}ll} \multicolumn{2}{@{}l}{\textbf{#1}}\\[3pt]
}{
  \end{tabular} \vskip 1em
}

\usepackage{nicefrac}

\title{On Dynamic Graph Algorithms with Predictions\thanks{This work is supported by the Austrian Science Fund (FWF): P 32863-N. This project has received funding from the European Research Council (ERC) under the European Union's Horizon 2020 research and innovation programme (grant agreement No~947702). Part of this work was performed when Yasamin Nazari was affiliated with University of Salzburg, and Adam Polak was affiliated with Max-Planck Institute of Informatics. The discussions leading to this work was initiated at the AlgPiE 2022 workshop, organized by IGAFIT. 
}}
\author{Jan van den Brand \\ {\small Georgia Institute of Technology} \and Sebastian Forster \\ {\small University of Salzburg} \and Yasamin Nazari  \\ {\small VU Amsterdam} \and Adam Polak \\ {\small Bocconi University}}
\date{}

\begin{document}

\pagenumbering{roman}

\maketitle

\begin{abstract}
    Dynamic algorithms operate on inputs undergoing updates, e.g., insertions or deletions of edges or vertices. After processing each update, the algorithm has to answer queries regarding the current state of the input data. We study dynamic algorithms in the model of algorithms with predictions (also known as learning-augmented algorithms). We assume the algorithm is given imperfect predictions regarding future updates, and we ask how such predictions can be used to improve the running time. In other words, we study the complexity of dynamic problems parameterized by the prediction accuracy. This can be seen as a model interpolating between classic online dynamic algorithms -- which know nothing about future updates -- and offline dynamic algorithms with the whole update sequence known upfront, which is similar to having perfect predictions. Our results give smooth tradeoffs between these two extreme settings.
 
Our first group of results is about partially dynamic problems with edge updates. We give algorithms for incremental and decremental transitive closure and approximate APSP that take as an additional input a predicted sequence of updates (edge insertions, or edge deletions, respectively). They preprocess it in $\tilde{O}(n^{(3+\omega)/2})$ time, and then handle updates in $\tilde{O}(1)$ worst-case time and queries in $\tilde{O}(\eta^2)$ worst-case time. Here $\eta$ is an error measure that can be bounded by the maximum difference between the predicted and actual insertion (deletion) time of an edge, i.e., by the $\ell_\infty$-error of the predictions.

The second group of results concerns fully dynamic problems with vertex updates, where the algorithm has access to a predicted sequence of the next $n$ updates. We show how to solve fully dynamic triangle detection, maximum matching, single-source reachability, and more, in $O(n^{\omega-1}+n\eta_i)$ worst-case update time.

Here $\eta_i$ denotes how much earlier the $i$-th update occurs than predicted.

Our last result is a reduction that transforms a worst-case incremental algorithm without predictions into a fully dynamic algorithm which is given a predicted deletion time for each element at the time of its insertion. As a consequence we can, e.g., maintain fully dynamic exact APSP with such predictions in
$\tilde{O}(n^2)$ worst-case vertex insertion time and $\tilde{O}(n^2 (1+\eta_i))$ worst-case vertex deletion time (for the prediction error $\eta_i$ defined as above).

Our algorithms from the first two groups, given sufficiently accurate predictions, achieve running times that go below known lower bounds for classic (without predictions) dynamic algorithms under the OMv Hypothesis. Moreover, our dependence on the prediction errors (so-called smoothness) is conditionally optimal, under plausible fine-grained complexity assumptions, at least in certain parameter regimes.

\end{abstract}

\newpage

\tableofcontents

\newpage

\pagenumbering{arabic}

\section{Introduction}

Dynamic algorithms maintain a solution to a computational problem -- e.g., single source distances in a graph -- for an input that undergoes a sequence of updates -- e.g., edge insertions or deletions. 
The goal is to process such updates as efficiently as possible, at least faster than recomputing the solution from scratch.

This is however not always plausible, as evidenced by numerous fine-grained conditional lower bounds, see, e.g., \cite{Patrascu10, AbboudW14, HenzingerKNS15}.

The recent line of research on learning-augmented algorithms provides many examples of how performance of classic algorithms can be provably improved using imperfect \emph{predictions}, generated, e.g., by machine-learning models (see surveys by Mitzenmacher and Vassilvitskii~\cite{MitzenmacherV20, MitzenmacherV22}). Among others, predictions allow us to improve competitive ratios of online algorithms (e.g., \cite{LykourisV21, PurohitSK18}), running times of static algorithms (e.g., \cite{DinitzILMV21, ChenSVZ22}), approximation ratios of polynomial-time approximation algorithms (e.g., \cite{ErgunFSWZ22, GamlathLNS22, NguyenCN23}). Often these improvements go beyond what is provably possible for classic algorithms (e.g., \cite{LykourisV21, ErgunFSWZ22} and many others). In this work we ask the following natural question:

\begin{center}
  \it How could we use predictions to improve dynamic algorithms?
\end{center}

We make two choices to narrow down this question. First, we focus on \emph{predictions about future input data}. This is akin to many previous results on learning-augmented online algorithms (e.g., \cite{LykourisV21, BamasMRS20}), and in contrast with settings where the predictions are about the output (e.g., \cite{AntoniadisCEPS23, DinitzILMV21, ErgunFSWZ22}). Second, we focus on \emph{improving the running time}, which is the most studied performance measure for dynamic algorithms. 

\paragraph{Online and Offline Dynamic Problems, and Conditional Lower Bounds.}

Dynamic problems are most often studied in their online variants, i.e., future updates are not known to the algorithm, and it has to perform them one by one.
On the other hand, offline dynamic algorithms (see, e.g, \cite{Eppstein94,SankowskiM10,KarczmarzL15,BringmannKN19,PengSS19,ChenGHPS20}) are given a sequence of updates upfront. Note that the offline model is equivalent to our proposed model with predictions in the case that predictions are perfectly accurate. That is, we study an interpolation between offline and online dynamic algorithms, and ask how an algorithm's performance degrades with increasing inaccuracies of predictions. This interpolation question makes sense only for dynamic problems whose offline variants admit faster algorithms than their corresponding online variants do. Ideally, we would like to be able to say that predictions -- even imperfect ones -- allow certain dynamic problems to be solved faster than what is provably plausible without such predictions in the classic online model. In other words, we aim for running times going below known conditional lower bounds.

There are many reductions showing hardness of dynamic problems under popular fine-grained complexity assumptions about static problems, such as 3SUM, APSP, or CNF-SAT (see, e.g., \cite{AbboudW14}). However, all these reductions share what is from our perspective a limitation: they are not adaptive, they are capable of producing the whole input sequence at once. Hence, they already imply hardness for offline variants of dynamic problems, and thus they cannot provide tight conditional lower bounds for those dynamic problems whose offline variants happen to be strictly easier than corresponding online variants. It is an interesting open problem to find a reduction -- from a natural static problem to a natural online dynamic problem -- that does not have this limitation~\cite{Abboud22}.

To our best knowledge, the only known tool in fine-grained complexity that is capable of distinguishing between online and offline variants of dynamic problems is the Online Matrix-Vector Multiplication (OMv) Hypothesis~\cite{HenzingerKNS15}, a conditional assumption about an online problem itself. Therefore, in this work we focus (mostly) on problems with known tight OMv lower bounds. These lower bounds often show that recomputing from scratch is basically the best we can hope for (without predictions).

\paragraph{Warm-up: OMv with Predictions.}

Before we delve into graph problems, let us start with a simple sanity check and verify that the cubic time barrier for OMv can be broken using predictions. If it was not the case, our project would be hopeless because then problems with known OMv lower bounds would remain hard even with predictions.

Recall that in the OMv problem we are first given Boolean matrix $M \in \{0, 1\}^{n \times n}$, and then we need to answer $n$ queries of the form: Given vector $v_i \in \{0, 1\}^n$, what is the Boolean product $Mv$? We have to answer queries one by one, in an online fashion: we have to output $Mv_i$ before we get to learn what $v_{i+1}$ is. OMv is conjectured to require cubic time, up to subpolynomial factors~\cite{HenzingerKNS15}. We propose the following variant of OMv with predictions:

\begin{problem}{Online matrix-vector multiplication (OMv) with predictions}
  Input offline:
    & matrix $M \in \{0,1\}^{n \times n}$; \\
    & predicted vectors $\hat{v}_1, \hat{v}_2, \ldots, \hat{v}_n \in \{0,1\}^n$. \\[3pt]
  Input online:
    & vectors $v_1, v_2, \ldots, v_n \in \{0,1\}^n$. \\[3pt]
  Output:
    & matrix-vector products $Mv_1, Mv_2, \ldots, Mv_n \in \{0,1\}^n$ over Boolean semiring.
\end{problem}

\noindent
We show a simple algorithm whose running time depends on the $\ell_1$-error of the predictions, which is one of the standard error measures for learning-augmented algorithms (see, e.g., \cite{LykourisV21, AntoniadisCEPS23}).

\begin{observation}
\label{obs:omvwpred}
OMv with predictions can be solved in total time $O(n^\omega + n\sum_i ||v_i-\hat{v}_i||_1)$. More precisely, preprocessing runs in $O(n^\omega)$ time, and $i$-th query requires $O(n||v_i-\hat{v}_i||_1)$ time.
\end{observation}

\begin{proof}
In the preprocessing phase, compute $M \cdot [\hat{v}_1, \hat{v}_2, \ldots, \hat{v}_n]$ over integers, in time $O(n^\omega)$. Then, after receiving each vector $v_i$, compute (over integers) $Mv_i = M\hat{v}_i + M(v_i - \hat{v}_i)$. The first term, $M\hat{v}_i$ can be retrieved in $O(n)$ time because it is the $i$-th column of the matrix computed during the preprocessing. The second term, $M(v_i - \hat{v}_i)$ can be computed in time $O(n ||v_i-\hat{v}_i||_1)$, because $||v_i-\hat{v}_i||_1$ equals the number of non-zeros of $v_i-\hat{v}_i$. Indeed, $||v_i-\hat{v}_i||_1 = ||\hat{v}_i-v_i||_0$, because $v_i$ and $\hat{v}_i$ contain only zeros and ones.
\end{proof}

Note that $||v_i-\hat{v}_i||_1 \leqslant n$, and hence $\sum_i ||v_i-\hat{v}_i||_1 \leqslant n^2$, so even with arbitrarily bad predictions our algorithm never needs more than cubic time, i.e., it is \emph{robust} in the learning-augmented terminology. On the other hand, for perfect predictions (i.e., $\sum_i ||v_i-\hat{v}_i||_1 = 0$) we achieve the running time of the best offline algorithm, i.e., we are \emph{consistent}. In \Cref{sec:discussion} we discuss these concepts further.

Moreover, the running time of this algorithm matches two conditional lower bounds. First, the dependence on the total prediction error $\eta = \sum_i ||v_i-\hat{v}_i||_1$ cannot be improved from $n\eta$ to $n\eta^{0.99}$ under the OMv hypothesis; this is because OMv reduces to OMv with predictions, with $\eta = O(n^2)$, by providing arbitrary (e.g., all-zero) vectors as predictions. Second, the $n^\omega$ term cannot be improved under the assumption that Boolean matrix multiplication is not easier than the general matrix multiplication, because even with perfect predictions (i.e., $\eta = 0$), solving OMv with predictions entails computing a Boolean matrix product. These arguments, however, do not rule out the possibility of, e.g., an $O(n^\omega + \eta^{3/2})$ time algorithm, which would be a meaningful improvement over the above algorithm. We find it an interesting open problem to provide a fine-grained conditional lower bound, matching the running time of \Cref{obs:omvwpred}, that holds already restricted to instances with $\eta = \Theta(n^\alpha)$ for arbitrary fixed $\alpha \in (0, 2)$.

Finally, note that going below the cubic time barrier for OMv, even with accurate predictions, requires fast matrix multiplication. Hence, it is not a surprise that our graph algorithms discussed below -- that go beyond known OMv-based lower bounds -- use algebraic algorithms for matrix multiplication.

\subsection{Our Results}

Our main results can be categorized into three groups that differ in the considered settings, types of predictions and measures of prediction errors. We discuss the groups one by one. \Cref{tab:results} provides a summary.

\paragraph{Partially Dynamic Graph Problems.}

Our first group of results is about partially dynamic problems with edge updates. These are problems on graphs undergoing edge insertions (incremental variant), and graphs undergoing edge deletions (decremental variant), but not both types of updates at the same time (this would be a fully dynamic variant, which we address later). While in general it is not the case that incremental and decremental variants must be equivalent, it turns out that all our results give the same bounds for both variants.

Our partially dynamic (incremental and decremental) algorithms take as an additional input a predicted sequence of updates (edge insertions or edge deletions, respectively). To illustrate our setting, we provide as an example a detailed definition of the incremental transitive closure problem, also known as all-pairs reachability, in directed graphs.

\begin{problem}{Incremental transitive closure with predictions}
  Input offline:
    & predicted sequence of edge insertions $\hat{e}_1, \hat{e}_2, \ldots, \hat{e}_m \in E$. \\[3pt]
  Input online:
    & sequence of interleaved updates and queries, i.e.,  \\
    & -- edge insertions $e_1, e_2, \ldots, e_m \in E$, and \\
    & -- reachability queries $(s_1, t_1), (s_2, t_2), \ldots, (s_q, t_q) \in V \times V$.\\[3pt]
  Output:
    & for each query $(s_i, t_i)$, \\
    & -- YES if there is a path from $s_i$ to $t_i$ in the graph at the current moment, \\
    & -- NO otherwise.
\end{problem}

The decremental variant is very similar; the difference is that the predicted sequence tells the algorithm in which order the edges are supposed to be deleted, not inserted, and that the online update operations are edge deletions, not edge insertions. In the incremental setting, it is natural to assume we start with an empty graph. However, in the decremental setting (also for classic algorithms without predictions), the initial graph has to be given upfront. In our setting with predicted sequence of updates, the initial graph can be considered to be given implicitly by the set of all edges appearing in the sequence. On the other hand, we can also think of an equivalent input format in which we are first given the initial graph $G=(V, E)$ where each edge $e \in E$ comes with an additional label $\hat{t}(e) \in [m]$ describing the predicted position of this edge in the sequence of deletions.

The second problem that we consider in this setting is the $(1+\epsilon)$-approximate all-pairs shortest paths (APSP) problem in unweighted directed graphs. The problem definition is similar to the above transitive closure definition -- the only difference is that each query $(s_i, t_i)$ has to output a number in range $[d(s_i, t_i), (1+\epsilon)\cdot d(s_i, t_i)]$, where $d(s_i, t_i)$ denotes the length of a shortest path from $s_i$ to $t_i$ at the current moment.

The query times of our algorithms for these problems depend on prediction error $\eta$ that has a somewhat technical definition. We discuss this in detail in \Cref{sec:overview:partially,sec:partially}. For now let us only say that $\eta$ can be upper bounded by the largest absolute difference between the predicted and actual insertion (deletion) time of an edge -- i.e., the $\ell_{\infty}$ error of the predictions -- which we denote by $\eta_{\infty} \defeq \max_{e \in E} |t(e) - \hat{t}(e)|$, where $t(e)$ denotes the actual update time of an edge (i.e., $t(e_i) \defeq i$), and $\hat{t}(e)$ denotes the predicted update time of an edge (i.e., $\hat{t}(\hat{e}_j \defeq j$). We discuss the choice of this error measure in \Cref{sec:discussion}.

Our partially dynamic algorithms with predictions are specified in \Cref{thm:intro:partially}.

\begin{theorem}[Full details in Theorems~\ref{thm:incrementalapsp} and~\ref{thm:decrementalapsps}]\label{thm:intro:partially}
Each of the following dynamic graph problems:
\begin{itemize}[noitemsep]
\item incremental transitive closure,
\item decremental transitive closure,
\item incremental $(1+\epsilon)$-approximate unweighted all-pairs shortest paths,
\item decremental $(1+\epsilon)$-approximate unweighted all-pairs shortest paths
\end{itemize}
with a predicted sequence of edge updates can be solved by a deterministic algorithm with  $\tilde{O}(n^{(3+\omega)/2})$ preprocessing time, $\tilde{O}(1)$ worst-case edge update time, and $\tilde{O}(\eta^2)$ query time, for prediction error $\eta \leqslant \eta_{\infty}$.
\end{theorem}

Without predictions, both these problems (transitive closure and approximate APSP) in both partially dynamic variants (incremental and decremental) can be solved by algorithms with $O(n^3)$ total update time and $O(1)$ query time~\cite{IbarakiK83,Italiano86,PoutreL87,DemetrescuI00,DemetrescuI06,BaswanaHS07,RodittyZ08,Lacki13,Bernstein13}, and under the OMv hypothesis this running time is tight (up to subpolynomial factors) even if one allows up to $O(n^{1-\varepsilon})$ query time~\cite{HenzingerKNS15}. In other words, classic algorithms without predictions face the cubic time barrier for these problems, and we show it can be bypassed with sufficiently accurate predictions.

We also show that the dependence of the running time in \Cref{thm:intro:partially} on the prediction error $\eta_{\infty}$ cannot be improved to strictly subquadratic under the OMv hypothesis.

\begin{restatable}{theorem}{ThmPartiallyOMvLB}
\label{thm:partiallyomvlb}
Unless the OMv Hypothesis fails, there is no algorithm for incremental (decremental) transitive closure with predictions with arbitrary polynomial prepocessing time, worst-case $O(n^{1-\varepsilon})$ edge insertion (deletion) time, and $O(\eta_\infty^{2-\varepsilon})$ query time, for any $\varepsilon > 0$. This already holds with respect to amortized time bounds per operation.
\end{restatable}

Finally, we note that the algorithms of \Cref{thm:intro:partially} are not \emph{robust} -- for large enough prediction errors they can be slower than the best known classic algorithm -- but they can be made robust, at least in the amortized-time sense, using the black-box approach described in \Cref{sec:discussion}.

\paragraph{Fully Dynamic Graph Problems with Predicted Vertex Updates.}

Here we consider fully dynamic variants of various graph problems on directed graphs with vertex updates: triangle detection, single source reachability, strong connectivity, cycle detection, maximum matching, number of vertex disjoint $st$-paths.
For \emph{online} vertex updates, there is no dynamic algorithm with $O(n^{2-\epsilon})$ update time for any constant $\epsilon>0$, conditional on the OMv hypothesis~\cite{HenzingerKNS15}.
All these problems (except triangle detection) can also be naively recomputed from scratch in $\tilde{O}(n^2)$ time \cite{BrandLNPSS0W20,ChenKLPGS22}. Thus no polynomial improvement over the trivial approach is possible for these online dynamic problems. 
For \emph{offline} vertex updates, there is no dynamic algorithm with $O(n^{\omega-1-\epsilon})$ update time, conditional on the triangle detection hypothesis~\cite{AbboudW14}.

On the upper bound side, this is matched by \cite{SankowskiM10,BrandNS19} with $O(n^{\omega-1})$ update time.
Our \Cref{thm:intro:algebraic} provides a smooth trade-off between the online and offline model.

We use $\eta_i$ to denote the prediction error of the $i$-th performed update. 
The prediction is a sequence of vertex updates. 
It can happen that an update occurs earlier than initially predicted, i.e.~an update is moved several positions ahead in the sequence. When the data structure performs the $i$-th update, $\eta_i$ is how many positions this update occurs too early. If this $i$-th update was not predicted at all (i.e.~it doesn't even occur in the predicted sequence), we define $\eta_i = \infty$.

\begin{theorem}[Full details in \Cref{thm:main_algebraic}]\label{thm:intro:algebraic}
    Fully dynamic triangle detection, single source reachability, strong connectivity, directed cycle detection, maximum matching size, number of vertex disjoint $st$-paths, with $\Theta(n)$ vertex updates and predictions can be solved in $O(n^{\omega} + n \sum_i \min\{\eta_i, n\})$ total time. 

    More precisely, we have $O(n^\omega)$ preprocessing time and the $i$-th update takes $O(n^{\omega-1}+n\cdot\min\{\eta_i, n\})$ time. 
\end{theorem}

We remark that one can interpret our total time complexity as scaling with the $\ell_1$-norm of the error.
For $\eta'_i := \min\{\eta_i, n\}$, $\eta'\in\mathbb{N}^n$ we have total time $O(n^\omega + n\|\eta'\|_1)$ for $n$ updates.

This result is obtained via reductions by \cite{Sankowski04,Sankowski07,BrandNS19}, which reduce these dynamic graph problems to dynamic matrix inverse.
In dynamic matrix inverse, we are given a dynamic $n\times n$ matrix $\mM$ and must maintain information about $\mM^{-1}$. The case of vertex updates on graph reduces to row and column updates to $\mM$ (i.e.~updates can replace one row and column of $\mM$ at a time).
The offline model with only column updates was studied previously in \cite{SankowskiM10,BrandNS19} and for entry updates in \cite{Kavitha14}.
The online model with row and column updates was studied in \cite{Sankowski04}.
We construct a dynamic matrix inverse algorithm with predictions in \Cref{sec:algerbraic_technical} which then implies \Cref{thm:intro:algebraic} via reductions from \cite{Sankowski04,Sankowski07}.

\paragraph{Fully Dynamic Graph Problems with Predicted Deletion Times.}

The last setting that we study is also a fully dynamic one but with a weaker prediction requirement than above.

First, only deletions are predicted; the algorithm needs no prior knowledge of insertions. Second, deletions are predicted only at the time of corresponding insertions. In other words, compared to a classic fully dynamic setting, the only difference is that each insertion comes with an additional number predicting when the currently inserted item (e.g., vertex or edge) is going to be deleted.

This model is inspired by a recent result by \cite{PR2023} in which they assume an offline sequence of predicted deletions -- i.e.\ deletion times have no error -- but in our case we can handle deletion errors. We first extend their techniques to give the following result:

\begin{restatable}{theorem}{thmPredictedDels}
\label{thm:semi_prediction}
    Consider a sequence of $T$ updates and suppose we are given an incremental dynamic algorithm with worst-case update time $\Gamma$. Assume also that at any point in time we have a prediction on the order of deletions of current items, such that for the $i$ inserted item the error $\eta_i$ indicates the number of elements predicted to be earlier than $i$-th item that actually arrive later ($\eta_i=0$ if the prediction is correct or the element arrives later). Then we have a fully-dynamic algorithm with the following guarantees:
    \begin{itemize}
    \item An insertion update is performed in $O(\Gamma \log^2 T)$ worst-cast time.
    \item The deletion of the $i$-th element can be performed in $O((1+ \eta_i) \cdot \Gamma \log^2 T )$ worst-case time. 
\end{itemize}
\end{restatable}

We can use this reduction, combined with an incremental APSP algorithm observed by \cite{Thorup2005} to get the following:

\begin{restatable}{theorem}{thmSemiAPSP}
\label{thm:apsp} Given a weighted and directed graph undergoing online vertex insertions and predicted vertex deletions, we can maintain exact weighted all-pairs shortest paths with the following guarantees:
\begin{itemize}
    \item An insertion update can be performed $O(n^2 \log^2 n)$ worst-cast time.
    \item A deletion of the $i$-th inserted vertex $v_i$ can be performed in $O( n^2 (\eta_i \log^2 n + 1))$ worst-case time, where error $\eta_i \in [0, n]$ indicates how many vertices were predicted to be deleted before $v_i$ that are actually deleted after $v_i$.
\end{itemize}
\end{restatable}

This can be compared to a recent fully dynamic worst-case exact APSP bound of $\tilde{O}(n^{2.5})$ by \cite{mao2023APSP} improving upon a long line of work on sub-cubic update times for APSP \cite{Thorup2005,abraham2017, gutenberg2020fully, chechik2023faster}.
We note that $n^{2.5}$ seems to be a natural barrier inherent to current algorithmic approaches for this problem, but there is no known conditional lower bound formalizing this intuition.

\begin{table}[t]
\centering
\caption{Summary of our results. Each column describes one group of results, built around one technique or data structure.}
\label{tab:results}

\vskip 1em

\small
\begin{tabular}{@{\hspace{2em}}l@{\hspace{4.25em}}l@{\hspace{4.25em}}l}
\toprule
\Cref{thm:intro:partially} & \Cref{thm:intro:algebraic} & \Cref{thm:semi_prediction} \\
\midrule
\multicolumn{3}{@{}l}{Setting:} \\
partially dynamic & fully dynamic & fully dynamic \\ 
(incremental or decremental) & (insertions and deletions) & (insertions and deletions) \\[0.5em]

\multicolumn{3}{@{}l}{Type of updates:} \\
edge updates & vertex updates & vertex updates \\[0.5em]

\multicolumn{3}{@{}l}{Predictions:} \\
sequence of all updates & sequence of next $n$ updates & deletion times \\
(insertions or deletions) &  & (given during insertions) \\[0.5em]

\multicolumn{3}{@{}l}{Running time:} \\
preprocessing: $\tilde{O}(n^{(3 + \omega)/2})$ & preprocessing: $O(n^\omega)$ & no preprocessing \\
update: $\tilde{O}(1)$ & update: & insertion: $\tilde{O}(n^2)$ \\
query: $\tilde{O}(\eta_\infty^2)$ & \hfill $O(n^{\omega-1}+ n \cdot \min\{\eta_i,n\})$ & deletion: $\tilde{O}(\eta_i n^2)$ \\[0.5em]
\multicolumn{3}{@{}l}{Error measure:}\\
$\ell_{\infty}$ & $\ell_1$ & $\ell_1$ \\[0.5em]

\multicolumn{3}{@{}l}{Applications to graph problems:} \\
transitive closure & triangle detection & exact APSP\\
$(1+\epsilon)$-approximate APSP & single-source reachability & \\
 & strong connectivity & \\
 & directed cycle detection & \\
 & maximum matching size & \\
 & \#vertex-disjoint $st$-paths & \\[0.5em]
 
\multicolumn{3}{@{}l}{Main technical tool:} \\
All-Pairs Bottleneck Paths & dynamic matrix inverse & reduction to incremental\\

\bottomrule
\end{tabular}
\end{table}

\subsection{Further Discussion}
\label{sec:discussion}

\paragraph{Consistency, Robustness, and Smoothness.}
Typically, algorithms with predictions are designed with three goals in mind: (1) \emph{consistency}, that is a near-optimal (or at least better than worst-case) performance when predictions are accurate; (2) \emph{robustness}, that is retaining worst-case guarantees of classic algorithms even when predictions are adversarial; and (3) \emph{smoothness}, that is a graceful degradation of algorithm's performance with increasing prediction error, providing an interpolation between the former two extremes. Let us discuss how these goals translate to our model of dynamic algorithms with predictions.

In this context, consistency alone is just equivalent to having an offline algorithm that is faster than the fastest known (or, even better, fastest conditionally possible) online algorithm.

Robustness can often be dealt with by black-box best-of-both-worlds types of arguments. It sometimes becomes an issues in contexts where the performance measure of choice is the competitive ratio, but it is rarely an issue when we optimize the running time. For static algorithms, one can just simulate two algorithms -- one with predictions, and another one with best known worst-case guarantees -- step by step, in parallel, and stop whenever one of these algorithms stops. This approach incurs only a factor-of-two multiplicative slowdown (which is negligible for asymptotic complexity) compared to the faster algorithm, on a per-instance basis. For dynamic algorithms, we need a more careful approach: Whenever the currently faster algorithm finishes processing a request, stop and return its answer; when a new request comes, resume the simulation from where it stopped, letting the slower algorithm possibly catch up. This way we can retain amortized running time guarantees of the better of the two algorithms, on a per-instance basis. We remark that it seems challenging to have a similar black-box tool for worst-case per request guarantees.

Smoothness is perhaps the least well defined of the three terms. Intuitively, we want the algorithms to tolerate as big prediction errors as possible without compromising on performance too much. For dynamic algorithms with predictions, some level of smoothness can always be achieved trivially. Assuming the best offline algorithm is polynomially faster than the best online algorithm -- which is anyway required to claim consistency -- one can always rerun the offline algorithm from scratch after encountering each difference between the predicted and the actual input sequence, and therefore tolerate some polynomial number of errors. We achieve better smoothness than this baseline benchmark by (1) incorporating error measures that distinguish between small errors and large errors, (2) getting better dependence on them, and (3) sometimes even showing that this dependence is conditionally optimal.

\paragraph{Predictions, Prediction Errors, and Learnability.}
We note that the predictions that we use, and error measures that quantify predictions accuracy, are standard in the learning-augmented literature. Predictions of the entire input sequence (\cref{thm:intro:partially}) are used, e.g., for scheduling problems~\cite{BamasMRS20}, predicting only a certain window of input sequence (\cref{thm:intro:algebraic}) is required for learning-augmented weighted caching~\cite{JiangPS22}, and predictions of the time of the next operation concerning the current item (\cref{thm:semi_prediction}) is the by-now-standard setup for unweighted caching~\cite{LykourisV21}.

The most ubiquitous way of measuring how for the prediction is from the truth is the $\ell_1$-distance (e.g., \cite{LykourisV21,DinitzILMV21,AntoniadisCEPS23}, and many more), but for certain problems (e.g., flow time minimization with uncertain processing times~\cite{AzarLT21,AzarLT22}) the $\ell_{\infty}$-error of predictions is a more natural (and sometimes even necessary) choice.

Our prediction errors can be illustrated with an example of road networks: Every day the same roads get congested during rush hour, so one can try to predict the updates in a dynamic road network. However, such predictions will not be perfect, because the exact order in which roads become congested may differ from day to day.

Since all our predictions are essentially permutations, the question of when such predictions can be efficiently learned is addressed by standard tools in the literature~\cite{HelmboldW09,KhodakBTV22}.

\subsection{Related Work}

Over the past couple of years the field of learning-augmented algorithms blossomed enormously, and it is implausible to list here all relevant contributions. We refer the interested reader to survey articles by Mitzenmacher and Vassilvitskii~\cite{MitzenmacherV20, MitzenmacherV22} and a website with a list of papers maintained by Lindermayr and Megow \cite{LindermayrM22}. There are numerous works on using predictions for improving competitive ratios of online problems (e.g., \cite{LykourisV21, PurohitSK18, BamasMRS20, AntoniadisCEPS23, AzarPT22}, and many, many more) and running times of static problems (e.g., \cite{DinitzILMV21, ChenSVZ22}).

Dynamic algorithms can be seen as a certain kind of data structures, and there are already several examples of learning-augmented data structures (see, e.g., \cite{KraskaBCDP18, FerraginaLV21, LinLW22}), but they focus primarily on index data structures, such as binary search trees, and hence they are not directly related to our work.

A concept related to offline dynamic graph algorithms is that of a graph timeline, as defined by~\cite{lacki2013reachability}, in which a sequence of graphs $G_1,\ldots,G_T$ is given upfront and any two subsequent graphs differ by only one edge being added or removed. \cite{KarczmarzL15} studies in this model several types of undirected connectivity queries over a time range, asking, e.g., if a path exists in at least one graph in a given interval.

Finally, there is a separate line of work on temporal graphs (see for instance the survey~\cite{holme2012temporal}), also related to the offline model, in which each edge is labelled with (a collection of) time intervals indicating when it is available. Often the goal in this line of work is understanding certain dynamics on networks (like information diffusion or convergence to certain properties), which is different from our computation efficiency objectives in the dynamic settings.

\paragraph{Concurrent Work.}
In the concurrent and independent work, Liu and Srinivas \cite{liu2023} consider the same predicted-deletions model as in our \Cref{thm:semi_prediction}. Their result is also based on a reduction from the fully-dynamic setting to the incremental setting. However, unlike our work, their algorithm does not directly rely on a similar reduction by \cite{PR2023}, whereas we use analysis of \cite{PR2023} as a black-box. We note that \cite{liu2023} present many other applications (e.g.~all-pairs max-flow/min-cut approximation, or uniform sparsest cut) in the predicted-deletions model that can also be derived from Theorem \ref{thm:semi_prediction}. 
On the technical side, they also show how to handle the case where number of updates $T$ is not known upfront, which we do not consider.

In another independent work, Henzinger, Saha, Seybold, and Ye \cite{henzinger2023}
initiate a systematic study of the time complexity of dynamic graph algorithms with predictions. While their focus is on conditional fine-grained lower bounds, they also provide some algorithms. In particular, their combinatorial (i.e., not using fast matrix multiplication) algorithms for transitive closure, approximate APSP, and triangle detection have bounds similar to our Theorems \ref{thm:intro:partially} and \ref{thm:intro:algebraic} but with worse preprocessing times. They also consider prediction models with error measures very different from ours.

\subsection{Notation}

We write $O(n^\omega)$ for the time complexity of multiplying two $n\times n$ matrices, where the current best bound is $\omega < 2.372$~\cite{DuanWZ23}. 

For the matrix product of rectangular matrices, we write $\MM(a,b,c)$ for the time complexity of multiplying $a\times b$ and $b\times c$ matrices.
We write $\mI$ for the identity matrix.

All the graphs considered in this paper are directed.

\section{Technical Overview}

In this section we briefly explain the main ideas behind our results.

\subsection{Partially Dynamic Algorithms (Section \ref{sec:partially})}
\label{sec:overview:partially}

Our algorithms for partially dynamic problems -- transitive closure, approximate APSP, exact SSSP -- use a connection between these problems and the all-pairs bottleneck paths (APBP) problem. The latter is a variant of the all-pairs shortest paths (APSP) problem in which, instead of minimizing the sum of edge weights, we minimize the maximum edge weight along a path. As opposed to APSP, which is conjectured to require cubic time~\cite{RodittyZ11,WilliamsW18}, APBP can be solved in strongly subcubic time~\cite{VassilevskaWY07}, and the best known APBP algorithm runs in $O(n^{(3+\omega)/2}) \leqslant O(n^{2.687})$ time~\cite{DuanP09}.

\paragraph{Transitive Closure.}
First, let us explain the connection of partially dynamic transitive closure with APBP. If we use edge insertion times as edge weights, and solve APBP, we obtain a matrix $B$ such that
\[B[u, v] = \min\bigl\{ \max_{e \in \mathcal P} \{\text{insertion time of } e\} \bigm\vert \mathcal{P} \in uv\textrm{-paths}\bigr\}.\]
Hence, $B[u, v] \leqslant k$ if and only if there is a path from $u$ to $v$ in the graph after the first $k$ insertions. This observation itself is sufficient to solve incremental\footnote{Or decremental! In the offline variant they are equivalent.} transitive closure in the offline setting (in other words, with perfect predictions) faster than the OMv-based cubic time lower bound for the online setting. 

Let us now explain how we handle prediction errors. After each edge insertion, we keep track of the longest prefix of the predicted sequence of updates that contains only the already inserted edges. Let us denote the length of this prefix by $p$. We also maintain the set of ``out-of-order'' edges $E_{err}$ that have been already inserted but are not contained in that prefix. Upon receiving a reachability query $(u, v)$ we construct an auxiliary graph $H$ on at most $2|E_{err}| + 2$ nodes: the endpoints of edges in $E_{err}$ and nodes $u$ and $v$. For every pair of nodes $x, y \in H$, we add an edge $(x,y)$ to $H$ if $(x,y) \in E_{err}$ or if $B[x, y] \leqslant p$. It is easy to see that there is a path from $u$ to $v$ in $H$ if and only if there is a path from $u$ to $v$ in the original graph. Constructing $H$ and finding a $uv$-path takes time $O(|E_{err}|^2)$, and we show that the number of out-of-order edges $|E_{err}|$ can be bounded by the $\ell_{\infty}$-error of the predictions. The same approach can be adapted to the decremental setting.

\paragraph{Approximate APSP.}
With an $O(\epsilon^{-1}\log n)$ overhead, in addition to answering queries on whether there is a path from $u$ to $v$, we are also able to report the length of a shortest path within up to $1+\epsilon$ multiplicative approximation error. Instead of the single matrix $B$, for every $d \in \{(1+\epsilon)^0, (1+\epsilon)^1, \ldots, (1+\epsilon)^{\log_{1+\epsilon}(n)}\}$, we compute matrix $B^{(d)}$ of bottleneck paths with up to $d$ hops. Each such matrix can be computed in $O(n^{(3+\omega)/2} \log d)$ time by repeatedly squaring the input weight matrix using $(\min,\max)$-product~\cite{DuanP09}. Note that $B^{(d)}[u, v] \leqslant k$ if and only if there is a path from $u$ to $v$ of length at most $d$ in the graph after the first $k$ insertions. Now, we can equip the auxiliary graph $H$ with edge weights corresponding to $(1+\epsilon)$-approximate distances in the original graph, and answer the queries by running the Dijkstra algorithm in $H$.

\subsection{Fully Dynamic Matrix Inverse with Predictions (Section \ref{sec:algerbraic_technical})}

By using standard reductions from dynamic graph problems to dynamic matrix inverse (see, e.g., \cite{Sankowski04,Sankowski07,BrandNS19}), \Cref{thm:main_algebraic} reduces to maintaining the matrix inverse of some matrix $\mM$ undergoing rank-1 updates, i.e.~updates where we are given two vectors $u,v$ and then set $\mM \leftarrow \mM+uv^\top$.
For these reductions, it suffices to return $v^\top\mM^{-1}$ after each update.

\paragraph{From Rank-1 to Entry Updates.}
In general, without predictions, rank-1 updates are strictly harder than entry updates (i.e.~updates that change only a single entry of the matrix $\mM$ at a time).
Rank-1 updates require $\Omega(n^2)$ update time \cite{HenzingerKNS15}, whereas entry updates can be handled in $O(n^{1.406})$ update time \cite{BrandNS19}.

However, in the prediction setting, one can actually reduce dynamic matrix inverse with rank-1 updates to dynamic matrix inverse with entry updates, i.e.~we can reduce updates for general dense $u$ and $v$ to the special case where $u$ and $v$ are sparse. (Entry updates are just the special case where $u,v$ have one non-zero entry each.) 

Let $(u^{(1)},v^{(1)}),\ldots,(u^{(n)},v^{(n)})$ be the next $n$ predicted rank-1 updates.

Then we can describe the rank-1 updates as follows. Let $\mU$ and $\mV$ be the $n\times n$ matrices obtained by stacking the vectors $(u^{(t)})_{t=1,\ldots,n}$ and $(v^{(t)})_{t=1,\ldots,n}$ next to each other. 
Let $\mD$ be a diagonal matrix that is initially all zero, and consider the matrix formula:
\begin{align}
f(\mM,\mU,\mV^\top,\mD) = \mV^\top(\mM + \mU\mD\mV^\top)^{-1}.
\label{eq:overview:formula}
\end{align}
Here, switching the diagonal entries of $\mD$ one-by-one from $0$ to $1$ corresponds to adding $u^{(t)} (v^{(t)})^\top$ to $\mM$ and then inverting the result.
Thus, the task of maintaining the inverse of $\mM$ subject to rank-1 updates, while returning $(v^{(t)})^\top \mM^{-1}$ after each update, can be reduced to the task of returning the $t$-th row of $f(\mA,\mU,\mV^\top,\mD)$ subject to entry updates to $\mD$.
In \cite{Brand21}, v.d.Brand has shown that any dynamic matrix formula that can be written using the basic matrix operations (addition, subtraction, multiplication and inversion), such as formula $f$ in \eqref{eq:overview:formula}, reduces to dynamic matrix inverse again (\Cref{lem:matrix_formula}), while supporting the same kind of updates and queries.

Thus, we must maintain the inverse of a certain matrix that is subject to \emph{entry} updates while supporting queries to its rows,
because the input to our formula $f$ only receives entry updates and we only require rows of $f$.
Since we only need to consider entry updates, that means we can now focus only on the special case where we receive rank-1 updates where both $u$ and $v$ are sparse with only one non-zero entry each.
Only if we receive an update that was not predicted at all, do we need to perform a rank-1 update with dense vectors.

\paragraph{Matrix Inverse with Predictions.}
Dynamic matrix inverse in the offline model where all updates are given ahead of time was solved in $O(n^\omega)$ total time by Sankowski and Mucha~\cite{SankowskiM10}, and later generalized by v.d.Brand, Nanongkai and Saranurak \cite{BrandNS19} to only require the sequence of column indices of all future updates but not the actual entries of the new columns.

These previous data structures are offline, i.e.~require correct predictions about the entire update sequence ahead of time and cannot support updates that differ from the prediction received during initialization. 
Building on their techniques, we construct a dynamic algorithm with predictions that is robust against inaccurate predictions.

Let $\mM$ be the dynamic input matrix. We write $\mM^{(i)}$ for a variant of $\mM$ that is updated only every $2^i$ iterations. So $\mM^{(0)}$ is always identical to $\mM$ and $\mM^{(i)}$ is identical to $\mM$ every $2^i$ iterations.

We maintain these matrices for $i=0,1,\ldots,\log n$ in the following implicit form. That is, only the matrices $\mL^{(i)},(\mR^{(i)})^\top \in \F^{n\times 2^i}$ are stored in memory where
\[
(\mM^{(i)})^{-1} = (\mM^{(i+1)})^{-1} (\mI + \mL^{(i)} \mR^{(i)}).
\]
The matrix $(\mM^{(\log n)})^{-1}$ is also stored explicitly in memory.
Note that maintaining $(\mM^{(\log n)})^{-1}$ takes $O(n^{\omega-1})$ amortized time, as we recompute this inverse in $O(n^\omega)$ time every $2^{\log n} = n$ updates.

Via the Woodbury matrix identity, one can show (\Cref{lem:woodbury_variant}) that the matrices $\mL^{(i)}$ and $\mR^{(i)}$ are of the form:
\[
\mR^{(i)}=(\mV^{(i)})^\top(\mM^{(i+1)})^{-1}
,\quad
\mL^{(i)}=\mU^{(i)}(\mI+\mU^{(i)}\mR^{(i)})^{-1},
\]
where $\mU^{(i)},\mV^{(i)}$ are given by the at most $2^i$ vectors $u,v$ of the past at most $2^i$ updates of the form $uv^\top$ by which $\mM^{(i)}$ and $\mM^{(i+1)}$ differ (these vectors will have one non-zero entry each, since we reduced to entry updates).
Here $\mR^{(i)}$ is composed of some (at most $2^i$ many) rows of $(\mM^{(i)})^{-1}$, because we have entry updates, and thus we can assume each $v$ to be a standard unit vector.
Therefore, we can compute $\mL^{(i)}$ and $\mR^{(i)}$ in $O((T_i+\MM(n,2^i,2^i))/2^i)$ amortized time,
where $T_i$ is the time required to obtain the at most $2^i$ rows of $(\mM^{(i+1)})^{-1}$.
If our predictions are correct, then when we previously computed $\mL^{(i+1)},\mR^{(i+1)}$, we could have also precomputed the required rows of $(\mM^{(i+1)})^{-1}$ in $O(\MM(n,2^{i+1},2^{i+1}))$ time, which is subsumed by the time required to compute $\mL^{(i+1)},\mR^{(i+1)}$.

Thus for correct predictions, we can assume $T_i=n2^i$.
This leads to $O(\sum_i \MM(n,2^i,2^i)/2^i) = O(n^{\omega-1})$ amortized time per update\footnote{We focus here in the outline on amortized complexity, but this can be made worst-case (see \Cref{sec:algerbraic_technical}).}.

A similar idea of maintaining $O(\log n)$ copies of matrix $\mM$ that are updated every $2^i$ iterations was also used in \cite{SankowskiM10,BrandNS19} but they did not handle incorrect predictions efficiently.

Now, observe what happens in our dynamic algorithm if a prediction is incorrect, i.e.~we perform an update in a column that was originally predicted to occur some $\eta$ iterations into the future. In that case the required rows of $(\mM^{(i+1)})^{-1}$ might not be precomputed.
However, the rows are precomputed in $(\mM^{(j)})^{-1}$ for every $j \le \min\{\log \eta, \log n\}$.
We can wlog assume that the rows are precomputed in $(\mM^{(\log n)})^{-1}$ because the entire inverse is computed from scratch every $n$ iterations.
Thus the missing row must only be computed in $(\mM^{(\ell)})^{-1}$ for $\ell = 0, 1, \ldots, \min\{\log \eta, \log n\}$.
So we obtain an additional $O(\sum_{\ell=0}^{\min\{\log \eta, \log n\}}n2^\ell) = O(n\min\{\eta,n\})$ cost for each update that occurs $\eta$ iterations earlier than initially predicted.

At last, consider what happens if we perform an update that was not predicted at all, i.e., we receive two vectors $u,v$ for a rank-1 update. Since the update was not predicted, the previous reduction does not hold and the vectors remain dense.
If $v$ is dense, computing the respective row of $\mR^{(i)}$ is not just copying a row of $(\mM^{(i+1)})^{-1}$ but rather it requires computing a vector-matrix product. Computing this product is done recursively, i.e.,
\[
v^\top (\mM^{(i)})^{-1} = v^\top (\mM^{(i+1)})^{-1}(\mI+\mL^{(i+1)}\mR^{(i+1)}) = v^\top (\mM^{(\log n)})^{-1} \prod_{\mathclap{j=i+1}}^{\mathclap{\log (n)-1}} (\mI + \mL^{(j)}\mR^{(j)}),
\]
which takes $O(n^2)$ operations. 
Note that for $i=0$, this actually computes $v^\top (\mM^{(j)})^{-1}$ for all $j=0,1,\ldots,\log n$ at once within $O(n^2)$ time.

\subsection{Fully Dynamic Algorithms with Predicted Deletion Times (Section \ref{sec:semi_predicted})}

We also consider the fully dynamic model in which predictions give no information about insertions whatsoever but the relative ordering of deletions is predicted -- by specifying for each item, at the time of its insertions, the position of its future deletion. Our algorithm is based on a result by Peng and Rubinstein \cite{PR2023}\footnote{A reduction similar to \cite{PR2023}, but with only an amortized update time guarantee, was also given by \cite{chan11}.} that gives a reduction from a \emph{fully dynamic semi-online} data structure, in which the order of deletions is given exactly, to an \emph{insert-only online} data structure. In particular, assuming that the insert-only data structure has \textit{worst-case} update time $\Gamma$, their semi-online data structure has update time $O(\Gamma \log T)$ for a sequence of $T$ updates.  We extend this reduction to the case where this order of deletions is not known exactly but it is predicted with some errors. The error for the $i$-th element is denoted by $\eta_i$, indicating that there are $\eta_i$ deletions that were predicted to happen before the deletion of element~$i$ but will arrive after its deletion. This error incurs an additional worst-case update time overhead of roughly $O(\eta_i \Gamma )$. 

At a high-level, the idea of Peng and Rubinstein \cite{PR2023} is that if the current list of the already performed insertions happens to be in the reverse order of deletion times, then a deletion can be performed in time $O(\Gamma)$ by simply rewinding the computation of the most recently performed insertion. Moreover, at any point in time, one can rewind some recent insertions and then re-insert these elements in a different order, to better prepare for future deletions. Since re-ordering the elements at each update would be expensive, they get an amortized bound by maintaining a sequence of buckets that keep partial reverse orderings. The amortized bound follows by ensuring that a set of $O(2^j)$ elements are re-ordered in every $2^j$ updates for each $ j = 0, \ldots, \lceil \log T \rceil $. In our case, when the deletion of the $i$-th element arrives $\eta_i$ positions \textit{earlier} than predicted, we rewind the computation of the last $\eta_i + 1$ insertions, until we get to delete the correct element, in time $O(\eta_i \Gamma)$, and then re-insert the $\eta_i$ unnecessarily deleted elements.

\section{Partially Dynamic Algorithms}
\label{sec:partially}

In this section we prove our upper bounds (\Cref{thm:intro:partially}) and lower bounds (\Cref{thm:partiallyomvlb}) for partially dynamic graph problems with predictions.

First, let us introduce two closely related concepts -- all-pairs bottleneck paths and $(\min,\max)$-product -- that we heavily use throughout the section.
In the all-pairs bottleneck paths (APBP) problem, we are given a directed graph $G=(V,E)$ with edge weights $w : E \to \mathbb{N}$, and we have to compute a matrix $B \in \mathbb{N}^{V \times V}$ with
\[B[u, v] \defeq \min\bigl\{ \max_{e \in \mathcal P} w(e) \bigm\vert \mathcal{P} \in uv\text{-paths}\bigr\}.\]
APBP can be solved in $O(n^{(3+\omega)/2}) \leqslant O(n^{2.687})$  time~\cite{DuanP09}.
The $(\min,\max)$-product of two $n \times n$ matrices $A$, $B$ is defined as $(A \circledvee B)[i,j] \defeq \min_k \max \{A[i,k], B[k, j]\}$. It can also be computed in $O(n^{(3+\omega)/2})$ time~\cite{DuanP09}.
Now, let us explain the relation between the two. Consider a directed graph $G=(V,E)$ with edge weights $w : E \to \mathbb{Z}$, and let $W$ denote the corresponding weight matrix, i.e., $W[u, v] = w(u, v)$ if $(u, v) \in E$, $W[u, v] = +\infty$ if $(u, v) \notin E$, and $W[u, u] = -\infty$. Observe that, for $d \in \mathbb{Z}_+$, the $(\min,\max)$-product of $d$ copies of $W$ gives all-pairs bottleneck paths with up to $d$ hops:
\[(\underbrace{W \circledvee W \circledvee \cdots \circledvee W}_{d\text{ times}})[u, v] = \min\bigl\{ \max_{e \in \mathcal P} w(e) \bigm\vert \mathcal{P} \in uv\text{-paths}, |\mathcal{P}| \leqslant d\bigr\}.\]
Such a product can be computed by the binary exponentiation

in $O(n^{(3+\omega)/2} \log d)$ time. For $d = n$, we get exactly APBP, and the extra $\log n$ factor can be avoided~\cite{VassilevskaWY07}.

\subsection{Upper Bounds}

Now we proceed to describe our partially dynamic algorithm for $(1+\epsilon)$-approximate APSP with predictions. Since transitive closure is a strictly easier problem, \Cref{thm:intro:partially} will follow. We begin with the incremental variant. Recall that $\hat{e}_1, \hat{e}_2, \ldots, \hat{e}_m$ denotes the predicted sequence of updates, and $e_1, e_2, \ldots, e_m$ the actual one.

\begin{theorem}
\label{thm:incrementalapsp}
Incremental $(1+\epsilon)$-approximate all-pairs shortest paths in unweighted directed graphs with predicted sequence of edge updates can be solved by a deterministic algorithm with $O(n^{(3+\omega)/2} \log^2 n)$ preprocessing time and $O(\log n)$ worst-case edge insertion time. Each query that is asked between the $i$-th and $(i+1)$-th insertion requires $O(\bar{\eta}_i^2 \log \log n) = O(\eta_{\infty}^2 \log \log n)$ time, where 
$\bar{\eta}_i$ is the current number of edges in the graph that are not contained in the longest prefix of the predicted sequence that has already been inserted, i.e., 
\[\bar{\eta}_i \defeq i - \max \big\{ j \mid \{\hat{e}_1, \hat{e}_2, \ldots, \hat{e}_j\} \subseteq \{e_1, e_2, \ldots, e_i \}\big\}.\]
\end{theorem}

Before proving the theorem, let us explain how the above measure of prediction error $\bar{\eta}_i$ can be upper bounded by the more standard $\ell_{\infty}$-error, denoted by $\eta_{\infty}$.
Let $\Pi \in S_m$ denote the permutation of predicted insertions corresponding to the actual sequence of insertions, i.e.,  $e_1, e_2, \ldots, e_m = \hat{e}_{\Pi(1)}, \hat{e}_{\Pi(2)}, \ldots, \hat{e}_{\Pi(m)}$. With this notation, we have $\bar{\eta}_i = i - \max \big\{ j \mid \{1, 2, \ldots, j\} \subseteq \{\Pi(1), \Pi(2), \ldots, \Pi(i) \}$.
Our goal is to show that $\bar{\eta}_i \leqslant \eta_{\infty} \defeq \max_j |j - \Pi(j)|$, for every $i \in [m]$. Fix $i \in [m]$, and let $k = \Pi^{-1}(i - \bar{\eta}_i + 1)$. Note that, by definition of $\bar{\eta}_i$, it holds that $i - \bar{\eta}_i + 1 \notin \{\Pi(1), \Pi(2), \ldots, \Pi(i)\}$. In other words,  $k \geqslant i + 1$. Then, $\eta_\infty \geqslant k - \Pi(k) \geqslant i + 1 - (i - \bar{\eta}_i + 1) = \bar{\eta}_i$, as desired.

\begin{proof}[Proof of \Cref{thm:incrementalapsp}]
Upon receiving the predicted sequence of edge insertions, the algorithm creates a weighted directed graph with edge set $E = \{\hat{e}_1, \hat{e}_2, \ldots, \hat{e}_m\}$ and edge weights $w : E \to \mathbb{Z}$ equal to predicted insertion times, i.e., $w\big(\hat{e}_i\big) = i$ for every $i \in [m]$.
Then, for every $d \in \{(1+\epsilon)^0, (1+\epsilon)^1, \ldots, (1+\epsilon)^{\log_{1+\epsilon}(n)}\}$, the algorithm computes matrix $B^{(d)}$ of bottleneck paths with up to $\lceil d \rceil$ hops. Observe that $B^{(d)}[u, v] \leqslant k$ if and only if there is a path from $u$ to $v$ of length at most $d$ using only edges from $\{\hat{e}_1, \hat{e}_2, \ldots, \hat{e}_k\}$. Computing all $B^{(d)}$'s takes $O(n^{(3+\omega)/2} \log^2 n)$ time in total.

On top of that, the algorithm creates a dictionary data structure (e.g., a balanced BST) that will allow translating edges given as pairs of nodes to their indices in the predicted sequence of insertions, and another BST that will maintain the set $S$ of indices of already inserted edges (initially, $S=\emptyset$. This ends the preprocessing phase.

When an edge $(u, v)$ is inserted, the algorithm first finds its index in the predicted sequence of insertions, i.e., $j$ such that $\hat{e}_j = (u,v)$, and then simply adds $j$ to $S$. This takes $O(\log n)$ time.

To handle a query $(u,v)$ the algorithm proceeds as follows. Let $i=|S|$ be the number of insertions so far. The algorithm first finds the largest prefix of the predicted sequence of insertions that has been already inserted, i.e., the largest $j$ such that $\{1, 2, \ldots, j\} \subseteq S$. This is simply the smallest positive integer not in $S$ minus one, and it can be found in $O(\log n)$ time assuming the BST maintains sizes and value ranges of its subtrees. Then, the algorithm uses the BST to list ``out-of-order'' edges $E_{err}$ that have been already inserted but are not contained in that prefix -- these correspond to elements of $S$ larger than $j$. In other words, the current edge set of the graph is exactly $\{\hat{e}_1, \hat{e}_2, \ldots, \hat{e}_j\} \cup E_{err}$. Note that $|E_{err}| = i - j = \bar{\eta}_i$, and $E_{err}$ can be constructed in $O(\bar{\eta}_i \log n)$ time.

Next, the algorithm creates a directed weighted auxiliary graph $H$. The nodes of $H$ are endpoints of edges in the list $E_{err}$ and nodes $u$ and $v$, that is at most $2\bar{\eta}_i + 2$ nodes in total. The edges of $H$ are of two kinds. First, there are all the edges from $E_{err}$, each with weight $1$. Second, for every pair of nodes $x, y \in H$, the algorithm selects the smallest $d$ such that $B^{(d)}[x, y] \leqslant j$, and, if such $d$ exists, it adds to $H$ edge $(x,y)$ with weight $d$. This second type of auxiliary edges represents paths using only edges from $\{\hat{e}_1, \hat{e}_2, \ldots, \hat{e}_j\}$, and their weights are upper bounds of the path lengths within up to $(1+\epsilon)$ multiplicative approximation error. It takes $O(\bar{\eta}_i^2 \log \log n)$ time to construct $H$.

Finally, the algorithm finds a (weighted) shortest path from $u$ to $v$ in $H$, which takes $O(\bar{\eta}_i^2)$ time, using the Dijkstra algorithm.
Let us justify that the (weighted) length of this path $d_H(u, v)$ is a correct $(1+\epsilon)$-approximation of the (unweighted) shortest path length in the original graph $d_G(u, v)$, i.e., that $d_H(u, v) \in [d_G(u, v), (1+\epsilon) \cdot d_G(u, v)] $. Clearly, $d_G(u, v) \leqslant d_H(u, v)$. For the remaining direction, fix a path of length $d_G(u, v)$ in the original graph. This path can be split into segments, each being either a single out-of-order edge or a subpath composed only of edges from $\{\hat{e}_1, \hat{e}_2, \ldots, \hat{e}_j\}$. Every such segment is represented by an edge in $H$ and the weight of this edge is at most $(1+\epsilon)$ times larger than the length of the segment. Hence, $d_H(u, v) \leqslant (1+\epsilon) \cdot d_G(u, v)$.
\end{proof}

Now we discuss the decremental variant, which is very similar. The main difference is that we will be looking at suffixes of the predicted sequence of deletions that were not yet deleted.

\begin{theorem}
\label{thm:decrementalapsps}
There is a deterministic decremental algorithm for $(1+\epsilon)$-approximate all-pairs shortest paths in unweighted directed graphs with predicted sequence of edge deletions with $O(n^{(3+\omega)/2} \log^2 n)$ preprocessing time and $O(\log n)$ worst-case edge deletion time. Each query that is asked between the $i$-th and $(i+1)$-st deletion requires $O(\bar{\bar{\eta}}_i^2 \log \log n) = O(\eta_{\infty}^2 \log \log n)$ time, where 
$\bar{\bar{\eta}}_i$ is the current number of edges in the graph that are not contained in the longest suffix of the predicted sequence that has not yet been deleted, i.e., 
\[\bar{\bar{\eta}}_i \defeq 
\min \big\{ j \mid \{\hat{e}_j, \hat{e}_{j+1}, \ldots, \hat{e}_m\} \cap \{e_1, e_2, \ldots, e_i \} = \emptyset \big\} - i - 1.\]
\end{theorem}

\begin{proof}
The algorithm closely mimics the incremental algorithm given in the proof of \Cref{thm:incrementalapsp}. We highlight the differences.

In the preprocessing phase the algorithm also computes hop-bounded bottleneck paths $B^{(d)}$'s for $O(\log n)$ exponentially growing hop bounds $d$, but the difference is that now the edge weight of edge $\hat{e}_i$ is $w(\hat{e})_i = -i$. It follows that $B^{(d)}[u, v] \leqslant -k$ if and only if there is a path from $u$ to $v$ of length at most $d$ using only edges from $\{\hat{e}_k, \hat{e}_{k+1}, \ldots, \hat{e}_m\}$.

Set $S$ still contains indices of edges present in the graph. That is, initially $S = [m]$, and deleting an edge boils down to removing its index from $S$.

To handle a query, the algorithm represents the current edge set of the graph $\{\hat{e}_j, \hat{e}_{j+1}, \ldots, \hat{e}_m\} \cup E_{err}$, for $j$ as small as possible. Note that $|E_{err}| = \bar{\bar{\eta}}_i$. The auxiliary graph $H$ again contains $u$, $v$, and endpoints of $E_{err}$. For $(x, y) \notin E_{err}$, the weight of $(x, y)$ in $H$ is the smallest $d$ such that $B^{(d)}[x, y] \leqslant -j$, if such $d$ exists and otherwise edge $(x,y)$ is not included in $H$. As before, $d_H(u, v) \in [d_G(u, v), (1+\epsilon) \cdot d_G(u, v)]$.

Observe that, as before, $\bar{\bar{\eta_i}} \leq \eta_{\infty}$.
\end{proof}

\subsection{Lower Bounds}
\label{sec:partialLB}
In this section we show that the dependence of the running time of our partially dynamic algorithms on the prediction error $\eta_{\infty}$ is (conditionally) optimal, at least in certain parameter regimes.

\ThmPartiallyOMvLB*

\begin{proof}
Henzinger et al.~\cite{HenzingerKNS15} proved that the OMv hypothesis implies that the following OuMv problem also cannot be solved in $O(n^{3-\varepsilon})$ time, for any $\varepsilon > 0$, even after arbitrary polynomial preprocessing time. In the OuMv problem we are first given Boolean matrix $M \in \{0, 1\}^{n \times n}$, and then we need to answer online $n$ queries of the form: Given two Boolean vectors $u_i, v_i \in \{0, 1\}^n$, what is the Boolean product $uMv$?

We show how to reduce an instance of OuMv to an instance of incremental (decremental) transitive closure with predictions, on a graph with $O(n)$ nodes, with a request sequence containing $O(n^2)$ updates and $O(n)$ queries, and with the maximum prediction error $\eta_{\infty} = O(n)$. The reduction itself runs in $O(n^2)$ time. Therefore, under the OMv hypothesis, it cannot hold simultaneously that the preprocessing time is polynomial in $n$, the update time is truly sublinear in $n$, and the query time is truly subquadratic in $\eta_{\infty}$.

We first focus on the incremental variant of the problem. We will think of the reduction as an algorithm solving the OuMv problem and having black-box access to an algorithm for incremental transitive closure with predictions. (We note that our reduction is modelled after a similar one in~\cite[Lemma 4.7 in the arXiv version]{HenzingerKNS15}, however we need to insert edges on two sides of the graph in order to reduce the number of queries and get a meaningful bound.)

Upon receiving matrix $M \in \{0,1\}^{n \times n}$, the reduction creates a graph composed of four layers of $n$ nodes each, and generates a predicted sequence of edge insertions. Let the vertex set be $V = \{a_1, \ldots, a_n\} \cup \{b_1, \ldots, b_n\} \cup \{c_1, \ldots, c_n\} \cup \{d_1, \ldots, d_n\}$, and let $E_M$ denote the following set of edges between $b$-nodes and $c$-nodes, corresponding to matrix $M$,
\[E_M = \{ (b_i, c_j) \mid (i, j) \in [n] \times [n], M[i,j] = 1\}.\]
The predicted sequence of edge insertions starts with all the edges from $E_M$, in an arbitrary fixed order, followed by
\begin{align*}
&(a_1, b_1), (c_1, d_1), (a_1, b_2), (c_2, d_1), \ldots, (a_1, b_n), (c_n, d_1),\\
&(a_2, b_1), (c_1, d_2), (a_2, b_2), (c_2, d_2), \ldots, (a_2, b_n), (c_n, d_2),\\
& \ldots, \\
&(a_n, b_1), (c_1, d_n), (a_n, b_2), (c_2, d_n), \ldots, (a_n, b_n), (c_n, d_n).
\end{align*}
The reduction gives this sequence to the algorithm to preprocess it, and then it inserts edges $E_M$, in the same order as in the sequence. This concludes the preprocessing phase, and the reduction starts accepting queries.

Upon receiving a pair of vectors $u_i, v_i \in \{0,1\}^{n \times n}$, the reduction first inserts edges from $a_i$ to $b$-nodes that correspond to ones in $u_i$, and from $c$-nodes that correspond to ones in $v_i$ to $d_i$, i.e.,
\[\{(a_i, b_j) \mid j \in [n], u_i[j] = 1\} \cup \{(c_j, d_i) \mid j \in [n], v_i[j] = 1\}.\]
At this point, the graph contains a path from $a_i$ to $d_i$ if and only if $uMv = 1$, so the reduction asks reachability query $(a_i, d_i)$ and returns the answer. Finally, the reduction inserts remaining edges from $a_i$ to $b$-nodes and from $c$-nodes to $d_i$, i.e.,
\[\{(a_i, b_j) \mid j \in [n], u_i[j] = 0\} \cup \{(c_j, d_i) \mid j \in [n], v_i[j] = 0\},\]
and it is ready to accept the next query.

Note that, both the predicted and actual insertion time for edge $(a_i, b_j)$ are within the range $[|E_m| + 2n(i-1), |E_m| + 2ni]$, so the prediction error for such edge is at most $2n$. The same is true for edges of the form $(c_j, d_i)$, and the predicted insertion times for edges between $b$-nodes and $c$-nodes have no error. Hence, the maximum prediction error is $\eta_{\infty} \leqslant 2n$, as desired.

The construction proving hardness of the decremental variant of the problem is very similar. The difference is that we start with two full bipartite cliques -- one between $a$-nodes and $b$-nodes, the other between $c$-nodes and $d$-nodes -- and edges $E_M$ between $b$-nodes and $c$-nodes. Then, in $i$-th OMv query, we first remove edges going from $a_i$ and to $d_i$ corresponding to zeros in $u_i$ and $v_i$, respectively; after that we ask the reachability query $(a_i, d_i)$, and finally we remove the remaining edges adjacent to $a_i$ and $d_i$, which correspond to ones in $u_i$ and $v_i$.
\end{proof}

\section{Dynamic Matrix Inverse with Predictions}\label{sec:algerbraic_technical}

In this section we prove \Cref{thm:matrix_inverse} which is our main algebraic data structure. 
In \Cref{sec:reductions} we use this result together with standard reductions from \cite{Sankowski04,Sankowski07,BrandNS19} to obtain the graph applications stated in \Cref{thm:main_algebraic}.
\begin{theorem}\label{thm:matrix_inverse}
    There exists a data structure with the following operations.
    \begin{itemize}
        \item \textsc{Initialize} Initialize on given $\mM\in\F^{n\times n}$ and a queue of $n$ rank-1 
        updates. Complexity $O(n^\omega)$.
        \item \textsc{AppendUpdate} Append 
        a rank-1 update (given via two vectors $u,v$) at the end of the queue in $O(n)$ worst-case update time.
    
        \item \textsc{PerformUpdate$(\eta)$} 
        Performs the update (i.e.~$\mM\leftarrow\mM+uv^\top$) stored at the $\eta$-th position in the queue, and removes it from the queue. 
        The data structure returns the rank and determinant of $\mM$.
        If the matrix is invertible, it also returns the vector $v^\top\mM'^{-1}$ (where $v$ is the vector of the performed rank-1 update and $\mM'$ is the matrix $\mM$ from before the update).
        The worst-case update time is $O(n^{\omega-1}+\min\{n\eta, n^2\})$.
    \end{itemize}
    The queue must have at least $n$ updates at all times.
    The data structure is randomized and its output is correct with high probability.
\end{theorem}
Note that here $\eta$ describes precisely the prediction error as described in the introduction, i.e.~the parameter $\eta$ describes how much earlier an update occurs than predicted.
If all updates occur exactly in the sequence as predicted, we always have $\eta = 1$. If an update occurs $\eta$ iterations too early, then it is stored at the $\eta$-th position in the queue.

\begin{remark}
If the matrix is promised to stay invertible throughout all updates, the data structure of \Cref{thm:matrix_inverse} can be deterministic. 
\end{remark}

\subsection{Reducing Rank-1 Updates to Column Updates}
We reduce the general rank-1 updates as described in \Cref{thm:matrix_inverse} to column updates which are easier to analyze.
Here by column update, we mean an update that changes only one column at a time.
Such a reduction is not possible in the general setting without predictions.
Rank-1 updates without predictions require $\Omega(n^2)$ update time under the OMv hypothesis, but column updates can be performed in $O(n^{1.529})$ time \cite{BrandNS19}.
However, since we are in the offline/prediction settings we can reduce the rank-1 updates to column updates.

The idea is as follows:
Given a set of $n$ predicted rank-1 updates $(u_i,v_i)_{i=1,\ldots,n}$, we can construct $\mU,\mV$ by stacking the vectors next to each other.
Then performing the rank-1 updates to some matrix $\mM$ could be phrased as follows: Let $\mD$ be an initially all-0 matrix and consider $\mM':=(\mM+\mU\mD\mV^\top)$.
By flipping the diagonal entries of $\mD$ from $0$ to $1$, the matrix $\mM'$ is precisely the matrix $\mM$ after receiving the rank-1 updates. In particular, if the first $k$ diagonal entries of $\mD$ are $1$, we have $\mM' = \mM + \sum_{i=1}^k u_i v^\top_i$.
Thus a rank-1 update to $\mM$ can be seen as a single entry update to $\mD$.

Further, it was shown that maintaining the value of any matrix formula $f(\mM_1,\ldots,\mM_k)$ that consists only of basic matrix operations (addition, subtraction, multiplication, inversion) can be reduced to a single matrix inversion.
That is, we can reduce the data structure task of maintaining the value of the formula
$$
f(\mM,\mU,\mV^\top,\mD) = (\mM+\mU\mD\mV^\top)^{-1}
$$
subject to entry updates to $\mD$, to a data structure that maintains the inverse of some matrix subject to entry updates
(and column updates are just a generalization of entry updates.).

Only if a rank-1 update was not predicted does an entry update to $\mD$ not suffice and we must perform an actual rank-1 update to $\mM$.

The following reduction is implicit from the following lemma by v.d.Brand \cite{Brand21}.

\begin{lemma}[{\cite{Brand21}}]\label{lem:matrix_formula}
    Given a matrix formula $f(\mA_1,\ldots,\mA_k)$ consisting of $p\ge k-1$ matrix operations,
    there exists a block matrix $\mB$ where some blocks are precisely $\mA_1,\ldots,\mA_k$ and the inverse $\mB^{-1}$ contains a block that is $f(\mA_1,\ldots,\mA_k)$. 
    When each $\mA_i$ is at most size $n\times n$, then $\mN$ is of size at most $O(pn)\times O(pn)$. 
    The proof is constructive and constructing $\mN$ takes time $O((pn)^2)$.

\end{lemma}

This reduction from predicted rank-1 updates to column updates motivates the following \Cref{lem:column_update_variant}, which can be interpreted as a restriction of \Cref{thm:matrix_inverse} to column updates.

\begin{restatable}{lemma}{columnupdatevariant}\label{lem:column_update_variant}
    There exists a data structure with the following operations.
    \begin{itemize}
        \item \textsc{Initialize} Initialize on given $\mM\in\F^{n\times n}$ and a queue of $n$ rank-1 
        updates in $O(n^\omega)$ time.
        \item \textsc{AppendUpdate} Append 
        a rank-1 update (given via two vectors $u,v$) at the end of the queue in $O(n)$ worst-case update time.

        \item \textsc{PerformUpdate$(\eta,\texttt{isQuery})$} 

        Performs the update (i.e.~$\mM\leftarrow\mM+uv^\top$) stored at the $\eta$-th position in the queue, and removes it from the queue. 
        The data structure returns the determinant of $\mM$, and $v^\top\mM^{-1}$ where $v$ is the vector of the performed rank-1 update.\\
        (We can decide to only perform a query, i.e.~return these values but do not change $\mM$.)
        Worst-case update time is $O(n^{\omega-1}+\min\{n\eta,n^2\})$ if $v$ was a standard unit-vector, otherwise it is $O(n^2)$.
    \end{itemize}
    The queue must have at least $n$ updates at all times.
    The data structure returns ``fail'' for the first time $\mM$ becomes singular. The data structure can no longer handle any updates after that point.
\end{restatable}

\subsection{Column Updates with Predictions}
\label{sec:reduction_rank1}

In this subsection, we prove \Cref{lem:column_update_variant}.
The main intermediate result is the following \Cref{lem:ranged_lookahead_matrix_inverse}.
At the end of this subsection we prove that \Cref{lem:ranged_lookahead_matrix_inverse} implies \Cref{lem:column_update_variant}.

\begin{lemma}\label{lem:ranged_lookahead_matrix_inverse}
    There is a data structure with the following operations:
    \begin{itemize}
        \item \textsc{Initialize} Initialize on given $\mM\in\F^{n\times n}$ and non-empty sets $F_0,\ldots,F_{\log n}\subset[n]$ where $|F_i|\le c\cdot2^{i+1}$ for $c \ge 1$ and $F_i\subset F_{i+1}$ in $O(n^\omega)$ time.
        These sets are predictions for the future column updates. Set $F_i$ contains the predicted column indices for the next $2^i$ column updates.
        \item \textsc{QueryAndUpdate} For $u,v\in\F^n$ return $v^\top \mM^{-1}$. Optionally, we can decide to set $\mM\leftarrow\mM+uv^\top$.

        If $v$ is some standard unit vector $e_j$ and $j\in F_i$ for some $i$, then this takes $O(cn^{\omega-1}+n2^i)$ time. 
        Otherwise it takes $O(cn^{\omega-1}+n^2)$ time.

        If this is the $t$-th update, then for all $\ell$ where $2^\ell$ divides $t$ the data structure must also receive new prediction sets $F_1,\ldots,F_\ell$ such that $F_i\subset F_{i+1}$ for all $i=0,\ldots,\log n$ and $|F_i|\subset c\cdot2^{i+1}$.
    \end{itemize}
    The matrix $\mM$ must stay invertible throughout all updates.
\end{lemma}

We briefly discuss why \Cref{lem:ranged_lookahead_matrix_inverse} will imply update complexities as stated in \Cref{lem:column_update_variant}.

The sets $(F_i)_{i = 0, \ldots, \log n}$ are the sets of possible column indices where we predict the future $2^i$ columns updates to be.
For instance, if we are promised a sequence of $n$ column updates with $j_1,\ldots,j_n$ being their column indices (i.e.~when given a queue of updates as in \Cref{lem:column_update_variant}), then we can pick $F_i =\{j_1,\ldots,j_{2^i}\}$ for $i=0,\ldots,\log n$.
However, note that we do not need such a precise prediction for our data structure to work. It is enough if we have some $O(2^i)$-sized prediction for the future $2^i$ updates.\footnote{This will later be crucial to maintain the rank of $\mM$, because the reduction from dynamic matrix rank to dynamic matrix inverse performs adaptive updates that cannot be accurately predicted.}

Assuming we have the promised sequence of $n$ column updates, but some column update happens $\eta$ iterations too early (i.e.~$\eta$ as in \textsc{PerformUpdate} in \Cref{thm:matrix_inverse}), then we can find that column index in some $F_k$ for $k\leqslant\min\{1+\log \eta, \log n\}$.
Thus such an update takes time $O(n^{\omega-1}+n2^k) =O(n^{\omega-1}+\min\{n\eta,n^2\})$ by \Cref{lem:ranged_lookahead_matrix_inverse}.
Note that w.l.o.g.~$F_{\log n} = [n]$, so for any column update, we can always assume $\eta \leqslant n$.
However, if an update is not a column update (i.e.~if we perform a general rank-1 update), the update complexity will be $O(n^2)$ according to \Cref{lem:ranged_lookahead_matrix_inverse}.

The proof of \Cref{lem:ranged_lookahead_matrix_inverse} relies on the following implicit representation of a matrix inverse.

\begin{lemma}\label{lem:woodbury_variant}
    Given $\mM\in\F^{n\times n}$ and a rank-$k$ update $\mU,\mV\in\F^{n\times k}$ we have
    \[
    (\mM+\mU\mV^\top)^{-1} = \mM^{-1} (\mI+\mL\mR),
    \]
    where $\mR = \mV^\top \mM^{-1}$, $\mL=\mU(\mI+\mU\mR)^{-1}$.
\end{lemma}

\begin{proof}
    The Woodbury identity \cite{Woodbury50} states that
    \begin{align*}
    (\mM+\mU\mV^\top)^{-1} 
    =&~ 
    \mM^{-1} - \mM^{-1}\mU(\mI+\mV^\top\mM^{-1}\mU)^{-1}\mV^\top\mM^{-1},
    \end{align*}
    which implies
    \begin{align*}
    (\mM+\mU\mV^\top)^{-1} 
    =&~
    \mM^{-1}(\mI-\mU(\mI+\mV^\top\mM^{-1}\mU)^{-1}\mV^\top\mM^{-1}) \\
    =&~
    \mM^{-1}(\mI-\mL\mR).
    \end{align*}
\end{proof}

\begin{proof}[Proof of \Cref{lem:ranged_lookahead_matrix_inverse}]

    We start by describing a data structure with \emph{amortized} time complexity and will later extend it to worst-case time.

    Let $\mM^{(i)}$ initially be the matrix $\mM$, but then we update $\mM^{(i)}\leftarrow \mM$ only every $2^i$ calls to \textsc{QueryAndUpdate}.
    So we always have $\mM^{(0)}=\mM$, but $\mM^{(i)}$ might be that status of $\mM$ some $2^i$ calls to \textsc{QueryAndUpdate} ago.
    Throughout all updates, we represent the inverses of these matrices in the following implicit form
    \begin{align}
        (\mM^{(i)})^{-1}=(\mM^{(i+1)})^{-1}(\mI+ \mL^{(i)} \mR^{(i)}) \label{eq:inverse_invariant}
    \end{align}
    where $\mL^{(i)},\mR^{(i)}\in\F^{n\times2^i}$.
    We do this for $i=0,\ldots,\log n$.
    \paragraph{Initialization.}
    We compute $\mM^{-1}$ and store this matrix as $\mM^{(\log n)}$. Then set $\mL^{(i)}=\mR^{(i)}=0$ for all $i$.
    \paragraph{Update.}
    Assume we receive the $t$-th call to \textsc{QueryAndUpdate} where $\ell$ is the largest integer such that $2^\ell$ divides $t$. We start describing the calculations performed by our data structure. Afterward we will analyze the complexity.
    
    First, we compute $v^\top(\mM^{(0)})^{-1}$ where $v$ is one of the vectors describing the rank-1 update given by the current call to \textsc{QueryAndUpdate},
    as we must return this result.

    Further, the data structure performs the following operations to update its internal representation of the inverse.

    To maintain invariant \eqref{eq:inverse_invariant} we must update $\mM^{(i)}$ for all $i\leqslant \ell$.
    Note that $\mM^{(\ell+1)}$ was last updated $2^\ell$ calls to \textsc{QueryAndUpdate} ago, so the difference between $\mM^{(\ell)}$ and $\mM^{(\ell+1)}$ are only the past $2^\ell$ updates.
    By \Cref{lem:woodbury_variant}, we can choose $\mL^{(\ell)}$ and $\mR^{(\ell)}$ as follows:
    
    Let $\mV$ be the $n\times 2^\ell$ matrix, where the columns are given by the vectors $v$ of the past $2^\ell$ calls to \textsc{QueryAndUpdate}.
    Then set $\mR^{(\ell)}=V^\top(\mM^{(\ell+1)})^{-1}$.
    
    Let $\mU$ be the $n\times 2^\ell$ matrix, where the columns are given by the vectors $u$ of the past $2^\ell$ calls to \textsc{QueryAndUpdate}.
    (Though all $u$ for which we only performed a query but no update, we will set the corresponding column in $\mU$ to $0$ instead.)
    Then set $\mL^{(\ell)}=\mU(\mI+\mU\mR)^{-1}$.

    Thus by \Cref{lem:woodbury_variant} we have
    $$
    (\mM^{(\ell)})^{-1}=(\mM^{(\ell+1)})^{-1}(\mI+ \mL^{(\ell)} \mR^{(\ell)}).
    $$
    For $i<\ell$, we set $\mL^{(i)}=\mR^{(i)}=0$ because $\mM^{(i)}=\mM^{(\ell)}$.
    In summary, we still satisfy \eqref{eq:inverse_invariant}.
    
    \paragraph{Complexity.} 
    For now, let us assume that we only have column updates,
    i.e.~we only change one column of $\mM$. In that case we can assume $v$ is a standard unit vector because we are adding $uv^\top=ue_j^\top$ to $\mM$ to reflect adding $u$ to the $j$-th column of $\mM$ for some $j$.

    Then computing $\mV^\top(\mM^{(\ell)})^{-1}$ is equivalent to picking $2^\ell$ rows of $(\mM^{(\ell)})^{-1}$. 
    Let's assume we can obtain all these rows within some time $T_\ell$.
    Computing $\mR^{(\ell)}$ now takes $O(\MM(n,2^\ell,2^\ell))$ time.
    So maintaining the representation \eqref{eq:inverse_invariant} for this specific $\ell$ takes $O((T_\ell+\MM(n,2^\ell,2^\ell))/2^\ell)$ amortized time.

    Now assume we have perfect predictions, i.e.~we whenever we update any $\mM^{(i)}$, we know precisely which future $2^i$ rows will be required of $(\mM^{(i)})^{-1}$.
    Then we could precompute these $2^i$ rows whenever we update $\mM^{(i)}$.
    In that case, $T_i=O(n2^i)$ since we just need to read the precomputed rows from memory.

    With this motivation, we always precompute the rows of $(\mM^{(\ell)})^{-1}$ with row index in $F_\ell$ whenever we update the representation \eqref{eq:inverse_invariant}.
    This costs an additional $O(c\MM(n,2^\ell,2^\ell))$ that amortizes over the next $2^\ell$ updates.
    (Note that computing rows of $(\mM^{(\ell)})^{-1}$ requires the same rows of $(\mM^{(\ell+1)})^{-1}$ but by $F_\ell\subset F_{\ell+1}$ these rows of $(\mM^{(\ell+1)})^{-1}$ have been precomputed already.)

    In summary, if all updates occur as predicted, i.e.~all future $2^i$ updates affect only the columns with index in $F_i$ for all $i=0,\ldots,\log n$, then the amortized update time would just be 
    \[O\big(\sum_{i=0}^{\log n} c\MM(n,2^i,2^i)/2^i)\big) = O(c\MM(n,2^{\log n},2^{\log n})/2^{\log n}) = O(cn^{\omega-1}).\]

    Now assume there is some error in our prediction, i.e.~the column index $j$ of an updated column is not in $F_i$ for some $i$. That is an issue since we need the $j$-th row of $(\mM^{(i)})^{-1}$ but that row was not precomputed. 
    To compute this row, we must spend $O(n2^i)$ time 
    (by dimension of $\mL^{(i)},\mR^{(i)}$) and must also compute the $j$-th row of $(\mM^{(i+1)})^{-1}$.   

    Thus by recursion, we spend $O(n2^k)$ time, where $k>i$ is the smallest integer such that $F_k$ contains  $j$\footnote{This parameter $k$ can be seen as the error in our prediction, i.e.~if an update happens $\eta$ iterations too early, then we will have $j\in F_k$ for $k<1+\log \eta$.},
    because then that row of $(\mM^{(k)})^{-1}$ was precomputed during a previous update.
    Thus our update time increases by an additive $O(n2^k)$.

    Now let us focus on the case where $v$ is not a standard unit vector, i.e.~we perform a general rank-1 update instead of a column update. 
    In that case, computing $v^\top(\mM^{(i)})^{-1}$ takes $O(n2^i)$ time plus the time to compute $v^\top(\mM^{(i+1)})^{-1}$.
    This leads to at most $O(n^2)$ time for $v^\top(\mM^{(\log n)})^{-1}$.

    Note that we must also return $v^\top(\mM^{(0)})^{-1}$ after each update. 
    This is subsumed by the $O(n2^k)$ and $O(n^2)$ cost above, depending on whether $v$ is a standard unit vector or not.
   
    \paragraph{Worst-Case.}
    The worst-case bounds are obtained via a standard technique. 
    The idea is as follows.
    When constructing $(\mM^{(\ell)})^{-1}$ (i.e.$(\mL^{(\ell)},\mR^{(\ell)}$), we spread the calculations over next $2^{\ell-1}$ calls to \textsc{QueryAndUpdate}.
    Thus updating \eqref{eq:inverse_invariant} for a specific $\ell$ introduces only $O(\MM(n,2^\ell,2^\ell)/2^\ell)$ worst-case cost per call to \textsc{QueryAndUpdate}.
    
    Note that this modification requires us to modify the recursion of \eqref{eq:inverse_invariant} a bit.
    $(\mM^{(\ell)})^{-1}$ is not immediately accessible as we spread its calculation over several updates. So $(\mM^{(\ell-1)})^{-1}$ cannot access $(\mM^{(\ell)})^{-1}$ yet.
    So instead, it will access the old variant of $(\mM^{(\ell)})^{-1}$ (the one we are currently replacing).
    This means that $\mR^{(\ell-1)},\mL^{(\ell-1)}$ must be larger by a factor of two ($2^\ell$ columns instead of $2^{\ell-1}$) because the old version of $(\mM^{(\ell)})^{-1}$ represents the matrix $\mM$ some $2^\ell$ updates ago.
    
    This old version of $(\mM^{(\ell)})^{-1}$ is accessible when we refresh $(\mM^{(\ell-1)})^{-1}$ because its computation was spread over $2^{\ell-1}$ updates.

    The same is done recursively. For any $i<\ell$, where $\ell$ is the largest integer that divides $t$ (during the $t$-th update), we set
    \[
    (\mM^{(i)})^{-1}=(\mM'^{(i+1)})^{-1}(\mI+\mL^{(i)}\mR^{(i)}),
    \]
    where $\mM'^{(i+1)}$ is either the old version of $\mM^{(i+1)}$, or the current version of it, if it has finished its calculation.

    Since the dimensions of the $\mL,\mR$ matrices increases only by a constant factor, the time complexity also increases by only a constant factor, but it is now worst-case.
\end{proof}

We can now prove \Cref{lem:column_update_variant} via \Cref{lem:ranged_lookahead_matrix_inverse}.

We restate \Cref{lem:column_update_variant}:
\columnupdatevariant*

\begin{proof}[Proof of \Cref{lem:column_update_variant}]
    Let us quickly recap some terminology.
    If $v$ of the rank-1 update is a standard unit vector, then the rank-1 update is a ``column update'' i.e.~it changes only one column. 
    We now define ``the column index of the update'', i.e.~we assign some index to each update:
    If the update is a column update, we select the index of the affected column. If the update is not a column update, we just assign $1$ as the column index.
    Thus ``the column index of the update'' is well-defined regardless of whether the update is a column or general rank-1 update.

    \paragraph{Maintaining sets $F_i$.}
    We have a queue of updates $Q$.
    This queue implies the sequence of sets $F_1,\ldots,F_{\log n}$ required by the data structure from \Cref{lem:ranged_lookahead_matrix_inverse}. 
    That is, at initialization (and every $2^i$ updates) $F_i$ contains the column-indices of the next $2^i$ updates for each $i=0,\ldots,\log n$.
  
    \paragraph{Initialization.}
    We are given the initial matrix $\mM$ and a queue of updates $Q$.
    This queue implies the sequence of sets $F_1,\ldots,F_{\log n}$ required by the data structure from \Cref{lem:ranged_lookahead_matrix_inverse}. 
    We initialize \Cref{lem:ranged_lookahead_matrix_inverse} on $\mM$ and $(F_i)_{i = 0, \ldots, \log n}$.
    \paragraph{AppendUpdate.}
    We are given a new update to append to the queue. We store this future update in $Q$.
    \paragraph{Update.}
    We are given a queue position $\eta$ and must perform the update stored at the $\eta$-th position in the queue. 
    If the update was a column update, then the column index of the update will be stored in some $F_i$ for $i \leqslant \min\{1+\log \eta, \log n\}$. Thus the update takes $O(n^{\omega-1}+\min\{n\eta, n^2\})$ operations.
    If this update is not a column update, then by \Cref{lem:ranged_lookahead_matrix_inverse} it takes $O(n^2)$ operations.
    The data structure of \Cref{lem:ranged_lookahead_matrix_inverse} returns $v^\top\mM^{-1}$.

    \paragraph{Determinant.}

    We have $\det(\mM+uv^\top) = \det(\mM)\cdot(1+v^\top\mM^{-1}u)$. Here $v^\top\mM^{-1}$ is given to us after each update, so we can maintain the determinant with an extra $O(n)$ overhead per update, which is subsumed by the cost of an update to \Cref{lem:ranged_lookahead_matrix_inverse}.
    We must compute $\det(\mM)$ during initialization which takes $O(n^\omega)$ operation, which is also subsumed by the initialization cost of \Cref{lem:ranged_lookahead_matrix_inverse}.
\end{proof}

\subsection{Putting Everything Together}

We now prove \Cref{thm:matrix_inverse} using \Cref{lem:column_update_variant} and the reduction from predicted rank-1 to column updates that we outlined in \Cref{sec:reduction_rank1}. 

\begin{proof}[Proof of \Cref{thm:matrix_inverse}]
    \Cref{thm:matrix_inverse} is almost the same as \Cref{lem:column_update_variant}, except that
    \begin{itemize}
        \item The update time $O(n^{\omega-1}+\min\{n\eta,n^2\})$ of \Cref{lem:column_update_variant} and \Cref{thm:matrix_inverse} matches only for column updates. Rank-1 updates are slower in \Cref{lem:column_update_variant} 
        with $O(n^2)$ update time.
        \item \Cref{thm:matrix_inverse} works on singular matrices and can maintain the rank.
    \end{itemize}

    We here describe how to extend \Cref{lem:column_update_variant} to obtain \Cref{thm:matrix_inverse}.

    \paragraph{Rank-1 Update Complexity.}

    We here describe how to reduce predicted rank-1 updates to column updates. Using this reduction, any predicted rank-1 update can be performed via a column update, we obtain $O(n^{\omega-1}+\min\{n\eta,n^2\})$ update time for predicted rank-1 updates.
    
    We can phrase the future rank-1 updates to $\mM$ as follows:
    Consider the formula $\mV'^\top(\mM+\mU\mD\mV^\top)^{-1}$. Here $\mU,\mV,\mV'$ are the vectors $u,v$ of the $n$ future rank-1 updates, but $\mV'$ contains one extra all-0 column, and $\mD$ is a diagonal matrix.
    Initially $\mD$ is all-0, and then one by one we set the diagonal entries to 1 to perform the queued rank 1 updates.
    Using \Cref{lem:matrix_formula}, there is a matrix $\mB$ that can be used to maintain the value of this formula.
    Since the formula consists of $5$ operations (1 addition, 3 products and 1 inversion), the matrix $\mB$ from \Cref{lem:matrix_formula} is of size $O(n)\times O(n)$ so there will be only a constant complexity blow-up.
    
    Any predicted rank-1 update requires us to just change one entry of $\mB$ (i.e.~the one entry of the block corresponding to $\mD$) so it can be performed via a column update to $\mB$ in $O(n^{\omega-1}+\min\{n\eta,n^2\})$ time via \Cref{thm:matrix_inverse}.
    A rank-1 update not already stored in $\mU,\mV$ takes $O(n^2)$ time (by \Cref{lem:ranged_lookahead_matrix_inverse}) because we perform a typical rank-1 update to $\mM$ (and thus a rank-1 update to $\mB$).
    After every update, we output one row of $\mV'^\top(\mM+\mU\mD\mV^\top)^{-1}$ which is contained in one row of $\mB^{-1}$ (and we know which row of $\mB$ that is, by \Cref{lem:matrix_formula} being constructive).
    In case of an update not stored in $\mV$ (and thus also not stored in $\mV'$), we set the extra 0 column of $\mV'$ to the vector $v$ of the rank-1 update. This, too, takes $O(n^2)$ time.
    
    We restart this dynamic algorithm every $n/2$ updates, thus $\mU,\mV$ will contain all updates that according to the queue should happen within the next $n/2$ iterations. This takes $O(n^{\omega-1})$ amortized time per update and can be made worst-case via standard techniques.
    By doing these restarts, any updates at position $\eta \leqslant n/2$ within the queue are stored in $\mU,\mV$ and thus have update time $O(n^{\omega-1}+n\eta)$.
    For updates at position $\eta > n/2$, they might not be stored in $\mU,\mV$ and thus are performed in $O(n^2)$ time. 
    
    In summary, we have $O(n^{\omega-1}+\min\{n\eta,n^2\})$ update time when performing an update stored at the $\eta$-th position within the queue.

\paragraph{Rank.}
    Sankowski \cite{Sankowski07} showed the following lemma:
    \begin{lemma}[{\cite{Sankowski07}}]\label{lem:rank}
        Given $\mM\in\F^{n\times n}$, let 
        \begin{align*}
            \mN := \begin{bmatrix}
                \mM & \mX & \\
                \mY &  & \mI \\
                & \mI & \mI_k
            \end{bmatrix} \in \F^{3n\times 3n}
        \end{align*}
        where $\mX,\mY$ are $n\times n$ matrices where each entry is a uniformly at random sampled number from $\F$. Matrix $\mI_k$ is the partial identity where only the first $k$ diagonal entries are $1$ and the remaining entries are $0$.
        Then with probability at least $1-O(n/|\F|)$ we have that $\mN$ is full rank if and only if $\rank(\mM) \ge n-k$.
        Further, if $k=0$, then the top-left $n\times n$ block of $\mN^{-1}$ is precisely $\mM^{-1}$.
    \end{lemma}
    We can reduce maintaining the rank of $\mM$ to maintaining the determinant of $\mN$ (as defined in \Cref{lem:rank}):
    Initially, compute the rank of $\mM$ and let $k=n-\rank(\mM)$. Then run the our data structure on $\mN$.
    With each update to $\mM$, the rank can change by at most $1$. If we observe that the rank decreases (because the determinant of $\mN$ becomes $0$), we unroll the last update, increase $k$, and then perform the original update again.
    Alternatively, if the determinant did not become $0$, we check if the rank increased by decreasing $k$ by $1$ (i.e. one extra update to $\mN$) and check if the determinant becomes $0$ (if so, revert this update again).
    Note that thus for each update to $\mM$, we will perform up to $2$ updates to $\mN$. 
    Further, if we look some $t$ updates to $\mM$ into the future, then the rank can change at most by $t$, so we have a range of size $O(t)$ where the future updates to $\mI_k$ will occur.
    So while we do not know the exact location of the future updates, we can still construct the sets $(F_i)_{0\le i\le \log n}$ as required by \Cref{lem:ranged_lookahead_matrix_inverse}: 
    For any $0\le i\le \log n$, let $F_i$ be the column indices of the next $2^i$ updates to $\mM$ and additionally the $2\cdot2^i$ column indices $2n+k+j$ for $-2^i \le j \le 2^i$ representing the range within which we may change the $\mI_k$ block of $\mN$.
    
    The failure probability of \Cref {lem:rank} is $O(n/|\F|)$. For most of our use-cases, we will have $|\F| = \poly(n)$. However, if $|\F|$ is not polynomial size, we can make the failure probability some small $n^{-c}$ for any arbitrary constant $c>0$ by instead using some field extension $\F'$ of polynomial size.
\end{proof}

\section{Fully Dynamic Graph Algorithms with Predictions}
\label{sec:reductions}

In this section we prove \Cref{thm:intro:algebraic}. We restate the result here in a more detailed way, i.e., we define the different operations of the data structure and how the predictions are given to it:

\begin{theorem}\label{thm:main_algebraic}
    There exists a fully dynamic algorithm that solves the following problems under vertex updates with predictions:
    triangle detection, single-source reachability, strong connectivity, directed cycle detection, maximum matching size, number of vertex disjoint $st$-paths. 

    The operations of the data structure are as follows
    \begin{itemize}
        \item \textsc{Initialize} Initialize on the graph and a queue of $n$ predicted vertex updates in $O(n^\omega)$ time.
        \item \textsc{AppendUpdate} Append a vertex update at the end of the queue in $O(n)$ worst-case time.
        \item \textsc{PerformUpdate$(\eta)$} 
        Performs the update stored at the $t$-th position in the queue, and removes it from the queue. 
        The worst-case update time is $O(n^{\omega-1}+n\min\{\eta, n\})$. 
     
    \end{itemize}
    The queue must have at least $n$ updates at all times.
   
\end{theorem}
\begin{remark}
    While we state the queue to need at least $n$ updates, any smaller $\Omega(n)$ also works by repeating each update $O(1)$ times.
\end{remark}
Note that $\eta$, the position of an update in the queue, as in \Cref{thm:main_algebraic} matches the definition of $\eta$ being an error measure of the prediction.
The queue can be seen as the predicted sequence of updates, and if all predictions are correct, we always perform the first update in queue, i.e.~$\eta=1$.
If the prediction is inaccurate and some update occurs $\eta$ iteration to early, then that update is stored not at the front of the queue, but at position $\eta$.

We now prove \Cref{thm:main_algebraic} via several reductions. Each subsection will present one reduction that solves one of the graph problems stated in \Cref{thm:main_algebraic} by reducing it to \Cref{thm:matrix_inverse}.

\subsection{Triangle Detection}
Given graph $G$ and its adjacency matrix $\mA$, the number of triangles in $G$ is given by $\sum_v (\mA^3)_{v,v}/3$, because $(\mA^3)_{v,v}$ is the number of paths from $v$ to $v$ using $3$ edges, i.e.~the number of triangles containing $v$.

After performing a vertex update to some $v\in V$, let $\mA'$ be the old adjacency matrix of $G$ and $\mA$ be the new one.
Then the number of triangles in $G$ changes by $(\mA^3)_{v,v}-(\mA'^3)_{v,v}$.
Thus we can maintain the number of triangles by querying only 2 entries of $\mA^3$ (one before and one after the update).

We can maintain $\mA^3$ via a matrix inverse by
$$
\begin{bmatrix}
\mI & \mA & & \\
& \mI & \mA & \\
& & \mI & \mA \\
& & & \mI
\end{bmatrix}^{-1}
=
\begin{bmatrix}
\mI & -\mA & \mA^2 & -\mA^3 \\
& \mI & -\mA & \mA^2 \\
& & \mI & -\mA \\
& & & \mI
\end{bmatrix}
$$
So we can solve triangle detection by running \Cref{thm:matrix_inverse} on the $3n\times3n$ matrix above.
Any vertex update to $G$ corresponds to updating one row and one column of $\mA$, so can be implemented via $3\cdot2=6$ rank-1 updates to the matrix above.

\subsection{Directed Cycle Detection}

\begin{lemma}[{\cite{BrandNS19}}]\label{lem:cycledirection}
    Let $G=(V,E)$ be a directed graph and $\mA \in \F^{n\times n}$ s.t.~each $\mA_{u,v}$ for $(u,v)\in E$ is picked independently and uniformly at random from $\F$, and all other entries of $\mA$ are $0$.
    Then with probability at least $1-n/|\F|$ we have $\det(\mI-\mA)=1$ if and only if $G$ is acyclic.
\end{lemma}
We run \Cref{thm:matrix_inverse} on matrix $\mI-\mA$ as in \Cref{lem:cycledirection}. A vertex update to $G$ corresponds to changing one row and column of $\mI-\mA$ so it can be performed with two rank-1 updates.

\subsection{Single Source Reachability and Strong Connectivity}

\begin{lemma}[{\cite{Sankowski04}}]\label{lem:reachability}
    Let $G=(V,E)$ be a directed graphs and $\mA \in \F^{n\times n}$ s.t.~each $\mA_{u,v}$ for $(u,v)\in E$ is picked independently and uniformly at random from $\F$, and all other entries of $\mA$ are $0$.
    Then for any $s,t\in V$ with probability at least $1-n/|\F|$ we have $(\mI-\mA)^{-1}_{s,t} \neq 0$ if and only if $s$ can reach $t$.
\end{lemma}
We run \Cref{thm:matrix_inverse} on matrix $\mI-\mA$ as in \Cref{lem:reachability}. A vertex update to $G$ corresponds to changing one row and column of $\mI-\mA$ so it can be performed with two rank-1 updates.
To return the single-source reachability for some source vertex $s$, we read the $s$-th row of $(\mI-\mA)^{-1}$.
This can be obtained by using \Cref{thm:matrix_inverse} to query $e_s^\top(\mI-\mA)^{-1}$.

To solve strong connectivity, not that a graph is strongly connected if and only if for any one vertex $v$, every other vertex can reach $v$ and $v$ can reach every other vertex.
Thus we can solve strong connectivity by running two data structures for single source reachability.

\subsection{Maximum Matching Size and Counting \texorpdfstring{$st$}{st}-Paths}

\begin{lemma}[{\cite{Lovasz79}}]\label{lem:matching}
    Given graph $G=(V,E)$ let $\mA\in\F^{n\times n}$ be the randomized Tutte matrix. That is, for each $(u,v)\in E$ let $\mA_{u,v}=-\mA_{v,u}$ be picked independently and uniformly at random from $\F$.
    Then for $|\F|=\Z_p$ ($p=\poly(n)$) we have w.h.p.~$\rank(\mA)=2\cdot$maximum matching size.
\end{lemma}
We run \Cref{thm:matrix_inverse} on the Tutte matrix $\mA$ as in \Cref{lem:matching}. A vertex update to $G$ corresponds to changing one row and column of $\mA$ so it can be performed with two rank-1 updates.

Counting the number of vertex disjoint $st$-paths can be solved via standard reduction to maximum bipartite matching size, see e.g.~\cite{MulmuleyVV87}.

\section{Fully Dynamic Algorithms with Predicted Deletion Times}
\label{sec:semi_predicted}

We consider the model in which insertions are arriving online but deletions are based on a predicted sequence, which we refer to as semi-online with prediction setting. We can extend the reduction of \cite{PR2023} that gives a reduction from a fully dynamic \textit{semi-online} data structure, in which the sequence of deletions are offline, to an \textit{insert-only} data structure. In particular, assuming that the insert only data structure has \textit{worst-case} update time $\Gamma$, the semi-online data structure of \cite{PR2023} has update time $O(\Gamma \log T)$ for a sequence of $T$ updates. We observe that an adaptation of their result can be used for the predicted deletion model. 

We start by sketching their amortized semi-online to worst-case insert-only reduction and then explain how this algorithm can be adapted to handle a deletion with error $\eta_i$. Specifically in the rest of this section we argue that the following theorem holds\footnote{Note that this claim holds for the problems in which the order in which elements are added does not impact the state of the problem -- which holds for almost all graph problems studied in the dynamic algorithms literature.}: 

\thmPredictedDels*

The key idea of the reduction in \cite{PR2023} is to order the list of elements in the reverse of deletion at any time and then perform each deletion by rewinding the computation (undo the insertion) in time $O(\Gamma)$ for the first element in this reversed list. Since re-ordering the elements at each update is expensive, they perform an amortization that performs the re-ordering partially for a set of $O(2^j)$ elements in every $2^j$ updates for each $ j = 0, \ldots, \lceil \log T \rceil $. 
In particular, they keep the elements in $L=\lceil \log T \rceil +1$ buckets $B_0, B_1,\ldots,B_{\lceil \log T \rceil}$ such that $B_j$ contains the elements indexed in range $[2^j, 2^{j+1})$ in the reverse order of deletion, and hence $B_0$ contains the next deletion to be performed. At a high-level, for each $j = 0, \ldots, L$, once in every $2^j$ updates, the algorithm re-orders a set of $O(2^j)$ elements (in $B_0 \cup \cdots \cup B_j$, whose total size is a sum over geometric-sized buckets) and rewinds the computation on these sets in time $O(\Gamma \cdot 2^j)$. The algorithm ensures that $B$ contains the first deletion that needs to be performed. Hence for each $j = 0, \ldots, L$ we get an amortized update time of $O(\Gamma)$. This amortization scheme is similar to the one we use in \Cref{sec:reductions}, with the difference that here the reverse ordered deletions in buckets are utilized and deletion is performed by a rewind (undo insertion) operation, which also takes $O(\Gamma)$ time in the RAM model. 

We get a similar reduction in the predicted deletion setting using the following adaptation: when the update (deletion) of the $i$-th element $e$ arrives $\eta:=\eta_i$ positions \textit{earlier} than predicted, we rewind the computation over all the sequence of elements in these $\eta$ positions until we get to the correct position of this element in time $O(\eta \Gamma)$ and then re-insert the $\eta-1$ deleted elements and update the lists as we would with any other insertion. In other words, we can simply maintain the state of the algorithm of \cite{PR2023} by performing $O(\eta)$ rewind operations and $O(\eta)$ re-insert operations, each of which takes $O(\eta \Gamma )$ time.
 
Note that if an element $e$ arrives later than predicted in the sequence, we simply ignore the deletion $e$ at the predicted arrival time and process it the actual arrival time. This will result in the next elements in the sequence being shifted earlier, and the cost will be incurred to the next elements in the sequence, so that when the element arrives the ordering is corrected by those earlier elements. 
 We can then run the algorithm of \cite{PR2023}, and note that the following invariants proven in \cite{PR2023} remain unchanged. 

For any time $t \in [T]$, let $\kappa(t)$ denote the largest integer such that $t$ is a multiple of $2^{\kappa(t)}$, let $E_t$ be the set of elements after the  $t$-th update and let $E_{t,r}$ be the set of $r$ elements that are deleted first (if there are less than $r$ elements, we set $E_{t,r}=E_t$). Their algorithm maintains the following invariants:
\begin{lemma}[\cite{PR2023}]\label{lem:semi_invariant_amortized}
    After $t$ updates the algorithm maintains:
    \begin{itemize}
        \item For each  $j= 0, \ldots, \kappa(t)$ we have $E_{t,2^{j+1}} \subseteq B_0 \cup \cdots \cup B_j$;
        \item $|B_0 \cup \cdots \cup B_{\kappa(t)+1}| \leqslant O(2^{\kappa(t)})$.  
    \end{itemize}
\end{lemma}

Hence, after performing the $O(\eta \Gamma)$ rewinds and re-insertion we can use Lemma \ref{lem:semi_invariant_amortized} to get the amortized update time over $T$ updates. At each time $t \in [T]$ we need to process the union of buckets $B_0 \cup .. \cup B_{\kappa(t)}$, and by summing over the bucket sizes we have:

\[ \frac{1}{T} \sum^T_{t=1} O(2^{\kappa(t)} \Gamma + \eta \Gamma) = O\big( (\frac{1}{T} \sum^{L}_{\kappa(t)=0} 2^{\kappa(t)} \cdot \frac{T}{2^{\kappa(t)}}) \Gamma + \eta L\Gamma\big)= O( (\eta+1)\Gamma \log T) \]

This algorithm is then de-amortized to get a worst-case bound with an additional log factor. At a high-level, the goal is to slowly perform the longer sequence of re-ordering and re-inserting operation over the updates. The challenge is that the upcoming insertions interleave the scheduled re-orderings. To handle this they distribute and preprocess the re-ordering tasks to $O(\log T)$ threads that further divide the updates into smaller geometrically decreasing sized \textit{epochs}. We refer the readers to~\cite{PR2023} for further details on the de-amortization details. They show that we can perform each update (without error) in worst-case $O(\Gamma \log^2 T )$ time. Our adaptation adds $O(\eta)$ rewind operations that take $O(\eta \Gamma)$ time, and introduces additional $O(\eta)$ insertions, which also take in total $O(\Gamma \log^2 T)$ time, and thus the worst-case time is $O(\Gamma \log^2 T (\eta+1))$.

\subsection{Application to All-Pairs Shortest Paths with Vertex Updates} 
We first observe that we can use the reduction of \cite{PR2023}, combined with an incremental APSP algorithm based on Floyd-Warshall, first observed by Thorup \cite{Thorup2005}, to get a fully dynamic semi-online APSP algorithm with $O(n^2 \log^2 n)$ worst-case update time. Note that here we can bound $T \leq n$ since at any point we can restrict our attention to a set of at most $n$ vertices inserted.

\begin{observation}[\cite{Thorup2005}]
Given an edge-weighted directed graph undergoing online vertex insertions, there is a deterministic algorithm that maintains exact all-pairs shortest paths in this graph with $O(n^2)$ worst-case update time.
\end{observation}

\begin{corollary}[\cite{PR2023, Thorup2005}] Given an edge-weighted directed graph undergoing online vertex insertions and offline (known) vertex deletions, there is a deterministic algorithm that maintains exact all-pairs shortest paths in this graph with $O(n^2 \log ^2 n)$ worst-case update time.
\end{corollary}

We can extend this to the setting in which the insertions are fully online, but the deletions are predicted, with the error measure described using Theorem \ref{thm:semi_prediction}.

\thmSemiAPSP*

\section{Acknowlegment}
The authors would like to thank Nicole Megow and Danupon Nanongkai for inspiring discussions on algorithms with predictions.
\newpage

\bibliographystyle{alphaurl}
\bibliography{ref}

\end{document}